\documentclass{article}
\usepackage{amsthm}

\usepackage{natbib}
\usepackage{graphicx,amsfonts,amsmath,amssymb,bm}
\usepackage{enumerate}
\usepackage{hyperref}
\usepackage{url}
\usepackage{xcolor}

\usepackage[a4paper,textwidth=5.85in]{geometry}
\usepackage{fullpage}

\def\balign#1\ealign{\begin{align}#1\end{align}}
\def\baligns#1\ealigns{\begin{align*}#1\end{align*}}
\def\bitemize#1\eitemize{\begin{itemize}#1\end{itemize}}
\def\benumerate#1\eenumerate{\begin{enumerate}#1\end{enumerate}}
\renewcommand{\exp}[1]{\operatorname{exp}\left(#1\right)} %
\providecommand{\argmin}{\mathop\mathrm{arg min}}

\def\E{\mathbb{E}} %

\def\E{\mathbb{E}} %

\newtheorem{example}{Example}
\newtheorem{assumption}{Assumption}

\newtheorem{corollary}{Corollary}
\newtheorem{theorem}{Theorem}
\newtheorem{lemma}{Lemma}

\begin{document}

\title{Model selection for estimation of causal parameters}
\author{Dominik Rothenh\"ausler
    }

\maketitle

\begin{abstract}
A popular technique for selecting and tuning machine learning estimators is cross-validation. Cross-validation evaluates overall model fit, usually in terms of predictive accuracy. In causal inference, the optimal choice of estimator depends not only on the fitted models, but also on assumptions the statistician is willing to make. In this case, the performance of different (potentially biased) estimators cannot be evaluated by checking overall model fit.

We propose a model selection procedure that estimates the squared $\ell_2$-deviation of a finite-dimensional estimator from its target. The procedure relies on knowing an asymptotically unbiased "benchmark estimator" of the parameter of interest. 
Under regularity conditions, we investigate bias and variance of the proposed criterion compared to competing procedures and derive a finite-sample bound for the excess risk compared to an oracle procedure. The resulting estimator is discontinuous and does not have a Gaussian limit distribution. Thus, standard asymptotic expansions do not apply. We derive asymptotically valid confidence intervals that take into account the model selection step. 
  
The performance of the approach for estimation and inference for average treatment effects is evaluated on simulated data sets, including experimental data, instrumental variables settings and observational data with selection on observables. 
\end{abstract}

\section{Introduction}

Model selection is a fundamental task in statistical practice. Usually, the aim is to find a model that optimizes  overall model fit. If the loss function is quadratic, the total error can be decomposed into the error due to variance and the error due to bias. %
A popular technique to balance the bias-variance trade-off in the context of prediction is cross-validation. However, the performance of common estimators in causal inference does not only depend on prediction performance, but, as we will discuss below, also on assumptions the statistician is willing to make. Thus, methods cannot be reliably compared by checking overall model fit. %
In the literature, there exist some solutions for parameter-specific model selection, but we currently lack a reliable general-purpose tool. In the context of causal inference, a reliable parameter-specific model selection tool could enable the following applications.

\paragraph*{Moving the goalpost.} Estimating average treatment effects from observational data  can be unreliable, if conditional treatment assignment probabilities are close to zero or one. To address this issue, some researchers move the goalpost by switching to causal contrasts that can be estimated more reliably, see \citet{lalonde1986evaluating}, \citet{heckman1998matching}, \citet{crump2006moving} and references therein. For example, instead of estimating the average treatment effect (ATE), a practitioner might switch to estimating the average treatment effect on the treated (ATT), or other causal contrasts such as the overlap effect. This is often done in an ad-hoc fashion. %
Using a model selection tool in this context could allow to trade off bias and variance when switching from estimating the ATE to estimators of other causal contrasts.  %

\paragraph*{Data fusion.} Combining evidence across data sets in the context of causal inference has recently attracted increasing interest \citep{peters2015causal,bareinboim2016causal, athey2020combining}. %
However, data quality often varies from data set to data set. In this case, using all data sets can lead to untrustworthy, biased estimators. A reliable model selection tool could make it possible to distinguish which data sets and estimators are useful for solving the estimation problem at hand, and which are not.

\paragraph*{Choosing between estimators that are optimal in different models.} If all confounders are observed, semi-parametrically efficient estimators such as augmented inverse probability weighting are attractive as they do not rely on parametric specifications \citep{robins1994estimation}. On the other hand, if the model is well-specified, parametric estimators such as ordinary least-squares can achieve higher efficiency and the fitted model may be easier to interpret. A reliable model selection procedure would allow practitioners to automatically select interpretable models if the parametric model is a good approximation, and semi-parametric estimators if the parametric model is far from the underlying data-generating process.

\paragraph*{Leveraging conditional independencies.}  

Researchers are often uncomfortable with using statistical procedures that only work under strong assumptions. Using such methods over other procedures may introduce some bias if the assumptions are violated, but has the potential to reduce variance. For example, conditional independence assumptions can be leveraged to improve precision of treatment effect estimators \citep{athey2019surrogate,guo2020efficient}.  A model selection tool as described above would allow to systematically trade off bias and variance when switching between estimation procedures that are optimal under different graphical structures. Doing so can potentially improve precision in scenarios in which researchers are not comfortable with making strong assumptions.  %

\subsection{Related work}

Model selection has a long history in statistics and machine learning. For optimizing loss-based estimators, the most commonly used methods include cross-validation, the  Akaike information criterion, and the Bayesian information criterion \citep{akaike1974new,schwarz1978estimating,friedman2001elements,arlot2010survey}.

The focused information criterion is a model selection criterion which, for a given focus parameter, estimates the mean-squared error of submodels \citep{claeskens2003focused,claeskens2008model}. It relies on knowing an asymptotically unbiased estimator of the parameter of interest. Its theoretical justification is given in a local misspecification framework. %

More recently, \citet{cui2019bias} introduce a model selection tool for finite-dimensional functionals in a semiparametric model if a doubly robust estimation function is available. It is based on a pseudo-risk criterion that has a robustness property if one of the estimators is biased.

For the task of model selection when estimating heterogeneous treatment effects, several methods have been developed \citep{kapelner2014inference,rolling2014model,athey2016recursive,nie2017quasi,zhao2017uplift,powers2018some}. Most of the methodologies are specific to the considered model class. A comparison of this line of work for individual treatment effects can be found in \citet{schuler2018comparison}.

\citet{van2003unified} propose a loss-based approach for parameter-specific model selection. In this work, the authors recommend minimizing an empirical estimate of the overall risk $ R(\hat \theta^{(g)}, \hat \eta)$,  %
where $\hat \theta^{(g)}$, $g=1,\ldots,G$ are candidate estimators and $\hat \eta$ is an efficient estimator of the nuisance parameter, computed on the training data set. %
Our approach is more generic in the sense that we do not assume that parameter of interest minimizes a known loss function.%

Closest to our work is the sample-splitting criterion developed by~\citet{brookhart2006semiparametric}. Roughly speaking, the data is split into a training and a test data set. Then, estimators are computed on the training and the test data set, and the squared deviation of estimators is aggregated across several splits.  The criterion developed by~\citet{brookhart2006semiparametric} can be seen as a form of Monte Carlo cross-validation. In the following, we discuss a variant of this approach that splits the data into $k$ folds and thus mimics $k$-fold cross-validation procedures, which are popular in practice. The data $D = (D_{1},\ldots,D_{n})$ is randomly split into $K$ disjoint  roughly equally-sized folds $D^{0,1},\ldots,D^{0,K}$. Define $D^{1,k} = D \setminus D^{0,k}$. Assuming that the data are i.i.d., $D^{1,k}$  and $D^{0,k}$ are independent for each $k$.
Let $\hat \theta^{(0)}$ be an unbiased estimator of the parameter of interest $\theta^{(0)} \in \mathbb{R}^d$. If several unbiased estimators are available, aggregation procedures such as inverse variance weighting can be used in a pre-processing step to obtain $\hat \theta^{(0)}$. Let $\hat \theta^{(g)}$ be candidate estimators, $g=0,\ldots,G$. Then, we can compute the risk criterion
\begin{equation}\label{eq:18}
  \tilde R(g) = \frac{1}{K} \sum_{k=1}^{K} \| \hat \theta^{(g)} (D^{1,k}) - \hat \theta^{(0)}(D^{0,k}) \|_2^{2}.
\end{equation}
Using independence of $D^{1,k}$ and $D^{0,k}$,
\begin{equation*}
  \mathbb{E}[\tilde R(g)] =  \mathbb{E}[ \| \hat \theta^{(g)}(D^{1,1}) - \theta^{(0)} \|_2^{2}]   + \sum_{j=1}^d \text{Var}(\hat \theta_j^{(0)}(D^{0,1})).
\end{equation*}
As $\text{Var}(\hat \theta^{(0)}(D^{0,1}))$ is constant in $g$, the criterion in equation~\eqref{eq:18} can be used to select an estimator $\hat \theta^{(g)}$ with low mean-squared error for estimating $\theta^{(0)}$ among $\hat \theta^{(g)}, g=0,\ldots,G$. The criterion in equation~\eqref{eq:18} is attractive as it is simple and widely applicable. We will compare the proposed model selection criterion to the criterion in equation~\eqref{eq:18}, both from a theoretical perspective and in simulations.

\subsection{Our contribution}

In this paper, we work towards making parameter-specific model selection more reliable. We derive a model selection criterion that estimates the squared $\ell_2$-deviation of an estimator from its target. We show that the selected model has equal or lower variance than the baseline estimator asymptotically. %
Compared to the model selection procedure \eqref{eq:18}, we show that the proposed criterion has equal or lower variance asymptotically. The proposed criterion is flexible in the sense that it can be used for any low-dimensional estimator in both parametric and semi-parametric settings. Theoretical justification of the method is given under the assumption of asymptotic linearity. 

Even if the candidate estimators are asymptotically linear, the model selection procedure is discontinuous and will not result in a regular estimator, even for $n \rightarrow \infty$. Thus, the final estimator does not have a Gaussian limit distribution. We derive asymptotically valid confidence intervals for the resulting estimator that takes into account the model selection step.

In previous work, the goal of model selection for estimation is usually to select a nuisance parameter model from a set of candidate models, but the identification strategy is held fixed across models \citep{van2003unified,brookhart2006semiparametric,cui2019bias}.  %
In this paper, the goal is to select among different estimands, or identification
strategies. Mathematically, this correspond to selecting among different \emph{functionals} of the underlying data generating distribution. In some situations, this results in dramatic improvements in the mean-squared error. However, there is no free lunch. Compared to the baseline procedure, for fixed $n$, model selection can lead to increased risk in parts of the parameter space.  We provide a finite-sample bound that reveals that the excess risk due to model selection becomes negligible as the dimension of the target parameter grows.

In simulations, the proposed criterion exhibits reliable performance in a variety of scenarios, including experiments, instrumental variables settings and data with selection on observables. The code can be found at \href{http://www.github.com/rothenhaeusler/tms}{github.com/rothenhaeusler/tms}. 

\subsubsection{Outline} %

In Section~\ref{sec:method-applications}, we introduce a method for parameter-specific model selection and discuss an example. Theory for the method is discussed in Section~\ref{sec:theory}. We evaluate the performance of the proposed procedure on simulated data in Section~\ref{sec:numerical-results}.

\section{Targeted model selection}\label{sec:method-applications}

This section consists of two parts. We briefly discuss the setting in Section~\ref{sec:setting-notation}. Then, we introduce the method in Section~\ref{sec:proposed-procedure} and discuss basic properties. %

\subsection{Setting and notation}\label{sec:setting-notation}

We observe data $D = (D_{i}$, $i=1,\ldots,n)$, where the $D_{i}$ are independently drawn from some unknown distribution $\mathbb{P}$. Suppose we have access to estimators $\hat \theta^{(g)}(D)$, $g=0,\ldots,G$, of some unknown parameter $\theta^{(0)}$. In the following, to simplify notation, we will write $\hat \theta^{(g)}$ instead of $\hat \theta^{(g)}(D)$. %
We assume that the baseline estimator $\hat \theta^{(0)}$ is asymptotically unbiased for $\theta^{(0)}$, i.e.\ that $\mathbb{E}[\hat \theta^{(0)}] = \theta^{(0)} + o(n^{-1/2})$. In practice, the data scientist may know several estimators that are asymptotically unbiased for the parameter of interest. In this case, one can use aggregation procedures such as inverse variance weighting to construct an optimally weighted aggregated estimator $\hat \theta^{(0)}$. 

In addition, the data scientist may have access to estimators $\hat \theta^{(g)}$ for which the data scientist is not sure whether they are approximately unbiased for the effect of interest. The goal is to select among the set of estimators, minimizing the mean-squared error with respect to the target of interest $\theta^{(0)}$. We assume that $\mathbb{E}[\hat \theta^{(g)}] = \theta^{(g)} + o(n^{-1/2})$ for some unknown $\theta^{(g)}$ and that $\sqrt{n}(\hat \theta^{(g)} - \theta^{(g)})$ converges to a non-degenerate random variable. We write $\sigma_j^{(g)}$ for the asymptotic standard deviation of $\sqrt{n}(\hat \theta^{(g)}_j - \theta^{(g)}_j)$. Similarly, we assume that $\sqrt{n}(\hat \theta^{(g)} - \hat \theta^{(0)} - (\theta^{(g)} - \theta^{(0)}))$ converges to a non-degenerate random variable for $g \neq 0$ and write $\tau_i^{(g)}$ for the asymptotic standard deviation of $\sqrt{n}(\hat \theta^{(g)}_j - \theta^{(g)}_j - \hat \theta^{(0)}_j + \theta^{(0)}_j)$. 

\subsection{The method}\label{sec:proposed-procedure}
We aim to find an estimator $g$ that minimizes
\begin{equation}\label{eq:1}
  R(g) = \mathbb{E}[ \| \hat \theta^{(g)} - \theta^{(0)} \|_2^{2}].
\end{equation}
Here and in the following, we suppress the dependence of $R(g)$ and $\hat \theta^{(g)}$ on $n$. As bias and variance of $\hat \theta^{(g)}$ are unknown, the function $R(g) $ is unknown and one has to minimize a proxy of the risk $R(g)$ instead. %
We propose to estimate $R(g)$ in equation~\eqref{eq:1} via
\begin{equation}\label{eq:13}
 \hat R(g) =   \| \hat \theta^{(g)} - \hat \theta^{(0)} \|_2^{2} + \sum_{j=1}^d  \frac{(\hat \sigma_j^{(g)})^2}{n} - \frac{ (\hat \tau_j^{(g)})^2}{n}, %
\end{equation}
where $\hat \sigma_j^{(g)}$ is an estimator of the asymptotic standard deviation of $\sqrt{n}(\hat \theta^{(g)}_j - \theta^{(g)}_j)$ and $ \hat\tau_j^{(g)}$ is an estimator of the asymptotic standard deviation of $\sqrt{n}(\hat \theta^{(g)}_j -\theta^{(g)}_j - \hat \theta^{(0)}_j + \theta^{(0)}_j)$. If the estimators are asymptotically linear (i.e.\ in some semi-parametric or low-dimensional parametric settings), consistent estimators $\hat \tau^{(g)}$ and ${\hat \sigma^{(g)}}$ are usually available via plug-in estimators of the variance of the influence function \citep{van2000asymptotic,tsiatis2007semiparametric}. An example will be discussed below. We propose to choose a final estimate $\hat \theta^{(\bar g)}$ by solving
\begin{equation}\label{eq:9}
\bar g =   \arg \min_{g} \hat R(g).
\end{equation}
 Let us consider a linear regression example. This example was mainly chosen for expository simplicity; the main motivating examples for this method are drawn from causal inference. The causal inference examples need more discussion and will be explained in detail in Section~\ref{sec:numerical-results}.
\begin{example}[Model selection for parameter estimation]
Usually, when conducting model selection in the context of prediction, the goal is to find a model that can be estimated well and is a good approximation of some complex model of interest. However, if the purpose is parameter estimation, fitting complex models can reduce variance while potentially introducing bias. Such settings appear in causal inference and will be further discussed in  Section~\ref{sec:numerical-results}. Here, we consider the task where the goal is to fit a regression with just one covariate; but there are additional covariates at our disposal that can be used to reduce variance, while potentially introducing some bias for the parameter of interest.
Let $ Y_i = X_i \theta^{(0)} + \epsilon_i$, where $D_i = (Y_i,X_i)$ are i.i.d.\ and the $\epsilon_i$ are noise terms that are uncorrelated of the $X_i$. Furthermore, for simplicity we assume that $Y_i$, $X_i$ and $\epsilon_i$ are centered. We are interested in the parameter $\theta^{(0)} = \arg \min \mathbb{E}[(Y - X \theta)^2]$ and consider the baseline estimator
\begin{equation*}
  \hat  \theta^{(0)} = \arg \min_\theta \sum_{i=1}^n (Y_i - X_i \theta)^2.
\end{equation*}
Let us assume that we have access to observations $Z_1,\ldots,Z_n$ from some additional covariate. One may consider the estimator
\begin{equation*}
  \hat \theta^{(1)} = \arg \min_{\theta} \min_\eta \sum_{i=1}^n (Y_i - X_i \theta - Z_i \eta)^2,
\end{equation*}
Let $(X,Y,Z,\epsilon)$ denote a generic $(X_i,Y_i,Z_i,\epsilon_i)$. 
If $Z$ is correlated with $Y$ and only weakly correlated with $X$, this estimator may reduce asymptotic variance compared to $\hat \theta^{(0)}$ since intuitively speaking, adjusting for $Z$ reduces unexplained variation in the residuals. On the other hand if $Z$ is strongly correlated with $X$, $\hat \theta^{(1)}$ may converge to a different parameter than $\hat \theta^{(0)}$.
Under regularity conditions \citep[Section 5]{van2000asymptotic}, $\hat \theta^{(0)}$ is asymptotically linear and unbiased for $\theta^{(0)} := \arg \min_\theta \mathbb{E}[(Y - X \theta)^2]$, i.e.\
\begin{equation*}
 \sqrt{n} (\hat \theta^{(0)} - \theta^{(0)}) = \frac{1}{\sqrt{n}} \sum_{i=1}^n \mathbb{E}[X^2]^{-1} X_i \epsilon_i + o_P(1). 
\end{equation*}
Similarly, under regularity conditions,
\begin{equation*}
  \sqrt{n} (\hat \theta^{(1)} - \theta^{(1)}) = \frac{1}{\sqrt{n}}\sum_{i=1}^n e_1^\intercal \mathbb{E}[(X,Z)^\intercal (X,Z)]^{-1}  (X_i,Z_i)^\intercal (Y_i - X_i \theta^{(1)} - Z_i \eta^{(1)})+ o_P(1),
\end{equation*}
where $(\theta^{(1)},\eta^{(1)}) = \arg \min_{(\theta,\eta)} \mathbb{E}[(Y - X \theta - Z \eta)^2]$ and where $e_j$ denotes the $j$-th unit vector.
Thus,
\begin{align*}
  (\sigma^{(0)})^2 &=  \text{Var}(  \mathbb{E}[X^2]^{-1} X \epsilon ),  \\
  (\sigma^{(1)})^2 &=  \text{Var}( e_1^\intercal \mathbb{E}[(X,Z)^\intercal (X,Z)]^{-1}  (X,Z)^\intercal (Y - X \theta^{(1)} - Z \eta^{(1)}) ,  \\
  (\tau^{(0)})^2 &=  0, \\
  (\tau^{(1)})^2 &=  \text{Var} \left( e_1^\intercal \mathbb{E}[(X,Z)^\intercal (X,Z)]^{-1}  (X,Z)^\intercal (Y - X \theta^{(1)} - Z \eta^{(1)}) - \mathbb{E}[X^2]^{-1} X \epsilon \right).
\end{align*}
These quantities can be consistently estimated via plug-in estimators in standard settings. For example, $(\sigma^{(1)})^2$ and $(\sigma^{(0)})^2$ can be consistently estimated via the sandwich estimator under regularity assumptions \citep{huber1967behavior}.
\end{example}

We will compare the risk proxy in equation~\eqref{eq:13} to sample-splitting based criteria in Section~\ref{sec:theory}. %
The method is evaluated on simulated data sets  in Section~\ref{sec:numerical-results}.%

\subsubsection{Improving precision} 

The risk criterion can be decomposed into several parts, i.e.\ $\hat R(g) = \sum_{j=1}^d \hat R_{\text{bias},j}(g) + \hat R_{\text{var},j}(g)$, where
\begin{equation*}
   \hat R_{\text{bias},j}(g) =  (\hat \theta_j^{(g)} - \hat \theta_j^{(0)})^2 - \frac{\hat \tau_j^{(g)}}{n},
\end{equation*}
and
\begin{equation*}
  \hat R_{\text{var},j}(g) =  \frac{\hat \sigma_j^{(g)}}{n}.
\end{equation*}
As the naming indicates, the first term can be interpreted as an estimate of the squared bias $( \theta^{(g)} - \theta^{(0)} )^2$, whereas the second term is an estimate of the variance of $\hat \theta^{(g)}$. Since we know that squared bias terms are non-negative this motivates defining the following modified risk criterion:
\begin{equation}\label{eq:10}
  \hat R^{\text{mod}}(g) =   \left( \sum_{j=1}^d  (\hat \theta^{(g)}_j - \hat \theta^{(0)}_j)^{2} - \frac{ ( \hat \tau_j^{(g)})^2}{n} \right)_{+} +  \sum_{j=1}^d\frac{(\hat \sigma_j^{(g)})^2}{n}.
\end{equation}
Then the final estimator $\hat \theta^{(\bar g)}$ is chosen such that $\bar g$ minimizes equation~\eqref{eq:10}. We take the positive part of the sum (instead of the sum of positive parts) as this allows random errors to cancel out for large $d$. This will be important for the theory developed in Section~\ref{sec:finite-sample-bound}. If there are ties, we select $\bar g$ as the one that minimizes $ \| \hat \theta^{(g)} - \hat \theta^{(0)} \|_2^{2}$ among the $g$ that satisfy $ \hat R^{\text{mod}}(g) = \min_{g'} \hat R^{\text{mod}}(g')$. The criterion $\hat R^{\text{mod}}(g)$ is not asymptotially unbiased for $R(g)$, but has some favorable statistical properties that we will discuss in the following section.

\section{Theory}\label{sec:theory}

In this section we discuss the theoretical underpinnings of the method introduced in Section~\ref{sec:method-applications}.  First, we show that the criterion $\hat R(g)$ is asymptotically unbiased for estimating the mean-squared error $R(g) =\mathbb{E}[ \|\hat \theta^{(g)} - \theta^{(0)} \|_2^{2}]$. Secondly, we compare the criterion $\hat R(g)$ to cross-validation in terms of asymptotic bias and variance. Then, we discuss the asymptotic risk of the resulting estimator.  We derive asymptotically valid confidence intervals for the parameter of interest that takes into account the model selection step. Finally, we present a finite-sample bound that shows that if the dimension of the target parameter is large, the excess risk due to model selection becomes negligible. %

\subsection{Assumptions}

We make two major assumptions, in addition to the assumptions outlined in Section~\ref{sec:setting-notation}. %
The first major assumption is a slightly stronger version of asymptotic linearity. Asymptotic linearity is an assumption that is commonly made  to justify asymptotic normality of an estimator \citep{van2000asymptotic,tsiatis2007semiparametric}. As our goal is to estimate the mean-squared error of an estimator, we use a slightly stronger version that guarantees convergence of second moments. The second major assumption is that the variance estimates are consistent. %

\begin{assumption}\label{assumptions} We make two major assumptions. 
\begin{enumerate}  
\item Let $\hat \theta^{(g)}$, $g=0,\ldots,G$ be estimators such that
\begin{equation*}
 \hat \theta^{(g)} - \theta^{(g)} =  \frac{1}{n} \sum_{i=1}^{n} \psi^{(g)}(D_{i}) + e_{g}(n),
\end{equation*}
where $\psi^{(g)}(D_{i})$ are centered and have finite nonzero second moments, and %
$\mathbb{E}[ \|e_{g}(n)\|_2^{2}] = o(1/n)$.  To avoid trivial special cases, in addition we assume that the covariance matrix of $(\psi^{(0)},\ldots,\psi^{(G)})$ is positive definite.%
\item The estimators of variance are consistent, that means
  \begin{align*}
  (\hat \tau^{(g)})^2   %
     &= (\tau^{(g)})^2 + o_{P}(1), \\
    (\hat \sigma^{(g)})^2  %
    &=  (\sigma^{(g)})^2 + o_{P}(1).
  \end{align*}
\end{enumerate}
\end{assumption} 
Let us compare the first part of the assumption to asymptotic linearity. Asymptotic linearity assumes that $\| e_{g}(n)^{2} \|_2^2 = o_{P}(1/n)$ while we assume that $\mathbb{E}[ \|e_{g}(n)^{2} \|_2^2] = o(1/n)$. Thus, our assumption is stronger than asymptotic linearity. Let us now turn to our theoretical results.
\subsection{Asymptotic unbiasedness}

Our first result shows that the proposed criterion is asymptotically unbiased for the mean-squared error of $\hat \theta^{(g)}$. The convergence rate depends on whether $\hat \theta^{(g)}$ is  asymptotically biased for estimating $ \theta^{(0)}$. If the estimator $\hat \theta^{(g)}$ is unbiased, the convergence rate is faster. The proof of the following result can be found in the supplement.

\begin{theorem}[Asymptotic unbiasedness of $\hat R(g)$]\label{theorem:unbiased}
Let Assumption~\ref{assumptions} hold.   \begin{enumerate}
\item    If $\theta^{(g)} = \theta^{(0)}$,
\begin{equation*}
 n ( \hat R(g) - \mathbb{E}[ \|\hat \theta^{(g)} - \theta^{(0)} \|_2^{2}])
\end{equation*}
converges weakly to a random variable with mean zero.
\item If $\theta^{(g)} \neq \theta^{(0)}$,
\begin{equation*}
 \sqrt{n} ( \hat R(g) - \mathbb{E}[ \|\hat \theta^{(g)} - \theta^{(0)} \|_2^{2}])
\end{equation*}
converges weakly to a random variable with mean zero. %
\end{enumerate}
\end{theorem}

This result shows that the estimator is asymptotically unbiased, which is important for the theoretical justification of the approach. %
The major strength of the proposed approach will become apparant in the next section, where we compare asymptotic bias and variance to the cross-validation criterion \eqref{eq:18}.

\subsection{Comparison to cross-validation}

In this section, we will show that the proposed criterion has asymptotically lower variance than the cross-validation criterion \eqref{eq:18} if $\theta^{(g)} = \theta^{(0)}$. Furthermore, we show that cross-validation is generally biased for estimating the mean-squared error if $\theta^{(g)} = \theta^{(0)}$. Let us first discuss asymptotic bias of the naive approach. The short proof of this result can be found in the supplement.

\begin{theorem}[Asymptotic biasedness of $\tilde R(g)$]\label{theorem:biased}
Let Assumption~\ref{assumptions} hold and fix $K$.
  \begin{enumerate}
\item    If $\theta^{(g)} = \theta^{(0)}$,
\begin{equation*}
 n \left(\tilde R(g) - \frac{K}{K-1} \mathbb{E}[ \| \hat \theta^{(g)} - \theta^{(0)} \|_2^{2}] - \frac{K}{n}\text{Var}(\psi^{(0)}(D_1)) \right)
\end{equation*}
converges weakly to a random variable with mean zero. %
\item If $\theta^{(g)} \neq \theta^{(0)}$,
\begin{equation*}
 \sqrt{n} (\tilde R(g) - \mathbb{E}[ \|\hat \theta^{(g)} - \theta^{(0)} \|_2^{2}] )
\end{equation*}
converges weakly to a random variable with mean zero. %
\end{enumerate}
\end{theorem}

Comparing this result with Theorem~\ref{theorem:unbiased}, we see a major difference how the methods behave when $\theta^{(g)} = \theta^{(0)}$. The proposed criterion is asymptotically unbiased for the mean-squared error, while cross-validation is not (even modulo additive constants). The cross-validation approach computes the estimator on randomly chosen subsets of the data. Thus, intuitively, this approach overestimates the asymptotic variance of unbiased estimators compared to the mean-squared error of biased estimators. In practice, this means that the approach tends to choose biased estimators with low variance over unbiased estimators with slighly higher variance. This effect can also be observed in Section~\ref{sec:numerical-results}. %

Now, let us turn to the asymptotic variance of the model selection criteria. The proof of the following result can be found in the supplement.

\begin{theorem}[Asymptotic variance of model selection criteria]\label{theorem:variance}
Let Assumption~\ref{assumptions} hold and fix $K$.
\begin{enumerate}
\item If $\theta^{(g)} = \theta^{(0)}$, the asymptotic variance of $n ( \hat R(g) - \mathbb{E}[\|\hat \theta^{(g)} - \theta^{(0)}\|_2^{2}] )^2$ is strictly lower than the asymptotic variance of $n (\tilde R(g) - \frac{K}{K-1} \mathbb{E}[ \| \hat \theta^{(g)} - \theta^{(0)} \|_2^{2}] - \frac{K}{n}\text{Var}(\psi^{(0)}(D_1)))$,
\item If $\theta^{(g)} \neq \theta^{(0)}$, the asymptotic variance of $\sqrt{n} ( \hat R(g)- \mathbb{E}[ \|\hat \theta^{(g)} - \theta^{(0)} \|_2^{2}])$ is equal to the asymptotic variance of $\sqrt{n} (\tilde R(g)- \mathbb{E}[ \|\hat \theta^{(g)} - \theta^{(0)} \|_2^{2}] )$.
\end{enumerate}  
\end{theorem}
Roughly speaking, this theorem shows that the proposed criterion $\hat R(g)$ has equal or lower asymptotic variance than the cross-validation criterion $\tilde R(g)$. This means that risk estimates based on the cross-validation criterion are harder to interpret than the risk estimates of the proposed criterion, since their variation will be larger. The difference in variance is particularly large if the splits are imbalanced, i.e.\ if the test data sets have much smaller or much larger sample size than the training data sets. Intuitively, if the validation data sets have small sample size, validation becomes unstable. However, if the validation data set is large, the estimator on the training data sets becomes unstable. This tradeoff can be avoided by estimating bias and variance separately, as done in the proposed approach. 

Note that for understanding the excess risk of the resulting estimator it is not sufficient to study the asymptotic variance of the criterion itself. Thus, we study the asymptotic risk in Section~\ref{sec:asymptoticrisk} and give a finite-sample bound for the excess risk in Section~\ref{sec:finite-sample-bound}.

\subsection{Asymptotic risk}\label{sec:asymptoticrisk}

Here and in the following, we focus on the case where the number of models $G$ is small and fixed and $n \rightarrow \infty$.  %
First, we investigate the asymptotic behaviour of the proposed procedure in the case where the number of models is fixed and $n \rightarrow \infty$. The proof of the following result can be found in the supplement. 
\begin{corollary}[Asymptotic risk of selected model]\label{cor:opt}
Let Assumption~\ref{assumptions} hold. Consider a finite and fixed number of estimators $g = 0, \ldots,G$. Let%
\begin{equation*}
   \bar g = \arg \min \hat R^{\text{mod}}(g).
\end{equation*} For $n \rightarrow \infty$,
  \begin{equation*}
     \mathbb{P}[ \theta^{(\bar g)} = \theta^{(0)}] \rightarrow 1,
   \end{equation*}
   and
   \begin{equation*}
     \mathbb{P}[ R(\bar g) \le R(0)] \rightarrow 1.
   \end{equation*}
 \end{corollary}
 In words, for $n \rightarrow \infty$, the proposed method selects models with lower or equal risk than the baseline estimator $\hat \theta^{(0)}$. Interestingly, an analogous result does not hold for the cross-validation procedure \eqref{eq:18}. In Section~\ref{sec:numerical-results} we will see that even for relatively large $n$, the selected model by cross-validation may have risk that is significantly larger than the risk of the baseline procedure. %

\subsection{Confidence intervals}\label{sec:infer-after-model}

Deriving confidence intervals that are valid in conjuction with a model selection step is a challenging topic and has attracted substantial interest in recent years, see for example \citet{berk2013valid,taylor2015statistical}. Generally speaking, statistical inference after a model selection step can be unreliable if the uncertainty induced by the model selection step is ignored. In this section, we describe how to construct confidence intervals that take into account the uncertainty induced by model selection. Intuitively speaking, the challenge is that the final estimator is a discontinuous function of the data. To be more precise, the final estimator $\hat \theta^{(\bar g)}$ is not regular and not asymptotically normal. 

The goal in this section is to find $I_1$ and $I_2$ as a function of the data $D_1,\ldots,D_n$ such that
\begin{equation*}
  \mathbb{P}[ \hat \theta_j^{(\bar g)} - I_1 \le \theta_j^{(0)} \le \hat \theta_j^{(\bar g)} + I_2] \rightarrow 1- \alpha,
\end{equation*}
for some pre-determined $\alpha > 0$ and $j$ and where
\begin{equation*}
  \bar g = \arg \min_g \hat R^{\text{mod}}(g).
\end{equation*}
The following theorem shows how to construct confidence intervals in low-dimensional settings; i.e.\ in settings where the number of models $G$ is fixed and the sample size goes to infinity.
\begin{theorem}\label{thm:ci}
 Define $\hat \theta = (\hat \theta^{(0)}, \hat \theta^{(1)},\ldots,\hat \theta^{(G)}) \in \mathbb{R}^{(G+1)d}$ and $\theta = (\theta^{(0)},\theta^{(1)},\ldots,\theta^{(G)})$. Assume that 
 \begin{equation*}
  \sqrt{n}(\hat \theta - \theta) \rightarrow \mathcal{N}(0,\Sigma),
\end{equation*} 
for some positive definite $\Sigma$. Let $\hat \Sigma(D) = \hat \Sigma_n(D_1,\ldots,D_n)$ be a consistent estimator of $\Sigma \in \mathbb{R}^{(G+1)d \times (G+1)d}$ and ${\hat \sigma^{(g)}} = {\hat \sigma^{(g)}}(D_1,\ldots,D_n)$ be a consistent estimator of the asymptotic standard deviation of $\sqrt{n}(\hat \theta^{(g)} - \theta^{(g)})$ and $\hat \tau^{(g)} = \hat \tau^{(g)}(D_1,\ldots,D_n)$ be a consistent estimator of the asymptotic standard deviation of $\sqrt{n}(\hat \theta^{(g)}- \hat \theta^{(0)} - \theta^{(g)} + \theta^{(0)})$. With some abuse of notation, conditionally on the data $D=(D_1,\ldots,D_n)$ draw $(Z_0,\ldots,Z_G) \sim \mathcal{N}(\sqrt{n} \hat \theta(D), \hat \Sigma(D))$ with $Z_g \in \mathbb{R}^d$. We define the event %
\begin{equation*}
  A_g^{est} = \{ \max( \| Z_g - Z_0 \|_2^2 - \|\hat \tau^{(g)} \|_2^2,0) + \|\hat \sigma^{(g)} \|_2^2 < \min_{g' \neq g} \max( \| Z_{g'} - Z_0 \|_2^2 - \|\hat \tau^{(g')}\|_2^2,0) + \|\hat \sigma^{(g')}\|_2^2 \}.
\end{equation*}
Now for some fixed $j$ define
\begin{equation}\label{eq:bootstrap}
  b_{D_1,\ldots,D_n}(\beta) = \sum_g \mathbb{P}[ \{ Z_{g,j} - \sqrt{n}\hat \theta_{j}^{(0)} \le \beta  \} \cap A_g^{est}| D_1,\ldots,D_n].
\end{equation}
Then:
\begin{enumerate}
  \item For $n \rightarrow \infty$, with probability converging to one, the inverse $b^{-1}_{D_1,\ldots,D_n} : (0,1) \rightarrow \mathbb{R}$ is well-defined.
  \item For all $\alpha > 0$,
  \begin{equation*}
      \mathbb{P} \left[ \hat \theta_j^{(\bar g)}  - \frac{ b^{-1}_{D_1,\ldots,D_n}(1-\alpha/2)}{\sqrt{n}} \le \theta_{j}^{(0)} \le \hat \theta_j^{(\bar g)} - \frac{ b^{-1}_{D_1,\ldots,D_n}(\alpha/2)}{\sqrt{n}} \right] \rightarrow 1-\alpha.
  \end{equation*}
\end{enumerate}
\end{theorem}
Note that by definition the conditional distribution of $Z$ given $D_1,\ldots,D_n$ is known to the researcher. Also, the researcher often can construct estimators of the variances in parametric and semi-parametric settings via plug-in estimators of the influence function, see e.g.\ \citep{van2000asymptotic,tsiatis2007semiparametric}. In these cases, $b_{i,\alpha}(D_1,\ldots,D_n)$ can be computed by the researcher, for example by Monte-Carlo simulation. Thus, Theorem~\ref{thm:ci} allows us to construct asymptotically valid confidence intervals for the final estimator $\hat \theta^{(\bar g)}$ in many parametric and semi-parametric settings.

This result is different from standard asymptotic arguments in the sense that $\hat \theta^{(\bar g)} - \theta^{(0)}$ is not asymptotically Gaussian. However, as the result shows, it is  possible to recover the exact asymptotic distribution of $\hat \theta^{(\bar g)}$ in low-dimensional scenarios and use this information to conduct asymptotically valid statistical inference. We will evaluate the empirical performance of confidence intervals constructed via Theorem~\ref{thm:ci} in Section~\ref{sec:numerical-results}.

\subsection{Finite-sample bound}\label{sec:finite-sample-bound}

As discussed in Section~\ref{sec:asymptoticrisk}, the proposed method selects a model that is asymptotically no worse than the baseline estimator. However, there is no free lunch. In transitional regimes, for fixed $n$, the estimator can perform worse than the baseline estimator $\hat \theta^{(0)}$. This is to be expected from statistical theory, see for example the discussion of the Hodges-Le~Cam estimator on page 110 in \citet{van2000asymptotic}. This makes it important to understand in which cases we can expect reliable performance of the proposed model selection procedure. In the case $d=1$, a Bayesian bound \citep{gill1995applications} reveals that improving over the Cram\'er-Rao bound in some parts of the parameter space must lead to deteriorating performance in other parts of the parameter space. In the following we will provide a finite-sample bound that will show that for large $d$ the excess risk becomes negligible, uniformly over a set of distributions.

 \begin{theorem}\label{thm:uniform}
    Let $\hat \theta^{(g)} = \theta^{(g)} + \epsilon^{(g)} =  \theta^{(0)} + \delta^{(g)} + \epsilon^{(g)}$ where $\delta^{(g)} \in \mathbb{R}^d$ is a constant vector and $G > 1$. We assume that the $\epsilon_j^{(g)}$ are centered, independent and sub-Gaussian random variables with variance proxy $\eta_j^{(g)}/\sqrt{n}$, i.e.\ that $\mathbb{E}[\exp{s \epsilon_j^{(g)}}] \le \exp{ \frac{(\eta_j^{(g)})^2 s^2}{2n}}$ for all $g=1,\ldots,G$ and $j=1,\ldots,d$ and $s \in \mathbb{R}$. %
    Define $\tau_j^{(g)}$ as the standard deviation of $\sqrt{n}(\epsilon_j^{(g)} - \epsilon_j^{(0)})$ and $\sigma_j^{(g)}$ as the standard deviation of $\sqrt{n} \epsilon_j^{(g)}$.
    We assume that $ |\delta_j^{(g)}| \le c_0/ \sqrt{n}$ for some constant $c_0 > 0$ and $\delta^{(0)} = 0$. Define $ \iota_{n,d} := \max(\sup_{g,j} |(\hat \sigma_j^{(g)})^2 - (\sigma_j^{(g)})^2 |,\sup_{g,j} | (\hat \tau_j^{(g)})^2 - (\tau_j^{(g)})^2 | )$. Furthermore, assume that $\log G/d \le c_1$ for some constant $c_1 > 0$. Then, for every $\kappa > 0$ there exists a constant $C$ that may depend on $c_0$, $c_1$, $\eta_j^{(g)}$ and $\kappa$ (but not on $\delta$) such that with probability exceeding $1-\kappa$,
\begin{equation*}   
 \frac{n}{d}  \sum_{j=1}^d (\hat \theta_{j}^{(\bar g)} - \theta_{j}^{(0)} )^2  \le \min_g \frac{n}{d}  \sum_{j=1}^d (\hat \theta_{j}^{(g)} - \theta_{j}^{(0)} )^2 +  C \sqrt{\frac{\log G}{d} } + 6\iota_{n,d}.
\end{equation*} 
 \end{theorem}

 In particular, the excess risk due to model selection goes to zero as $ \frac{d}{\log G} \rightarrow \infty$ and if $\sup_{g,j} |(\hat \sigma_j^{(g)})^2 - \sigma_j^{(g)}| \rightarrow 0$ and $\sup_{g,j} |\hat \tau_j^{(g)} - \tau_j^{(g)}| \rightarrow 0$. In Section~\ref{sec:numerical-results} we will see in an example that the excess risk declines for growing dimension of the estimand.

\section{Applications}\label{sec:numerical-results}

In this section, we discuss applications of the proposed method. Compared to existing literature on model selection for causal effects, instead of selecting among nuisance parameter models, we consider shrinking between different functionals of the data generating distribution. As we will see, doing so can lead to drastic improvements in the mean-squared error. However, there is no free lunch. Compared to the baseline procedure, for fixed $n$, model selection can lead to increased risk in parts of the parameter space. Thus, in this section, we study the excess risk of the procedures across the parameter space.

In the following we will use potential outcomes to define causal effects \citep{rubin1974estimating,splawa1990application}. We are interested in the causal effect of a treatment $T \in \{0,1\}$ on an outcome $Y$. Let $Y(1)$ denote the potential outcome under treatment $T=1$ and $Y(0)$ the potential outcome under treatment $T=0$. We assume a superpopulation model, i.e.\ $Y(1)$ and $Y(0)$ are random variables. In the following, the goal is to estimate the average treatment effect within several subgroups, %
\begin{equation}\label{eq:5}
 \theta_s^{(0)} = \mathbb{E}[Y(1) - Y(0)|S=s].
\end{equation}
Many methods have been designed to estimate \eqref{eq:5} and these methods operate under a variety of assumptions. We present several applications that are based on different sets of assumptions for identifying \eqref{eq:5}. In each of the cases, we compare the proposed method (\ref{eq:10}, termed ``targeted selection'') with the cross-validation procedure \eqref{eq:18}  and with a baseline estimator. The code can be found at \href{http://www.github.com/rothenhaeusler/tms}{github.com/rothenhaeusler/tms}.

 \subsection{Observational studies}

 In observational studies, it is common practice to estimate causal effects  under the assumption of unconfoundedness and under the overlap assumption. Roughly speaking, the overlap assumption states that treatment assignment probabilities are bounded away from zero and one, conditional on covariates $X$. If these assumptions are met, it is possible to identify the average treatment effect via matching, inverse probability weighting, regression adjustment, or doubly robust methods \citep{hernan2010causal,imbens2009recent}. However, if the overlap is limited, estimating the average treatment effect can be unreliable. %

 To deal with the issue of limited overlap, researchers sometimes switch to different estimands such as the average effect on the treated (ATT) or the overlap weighted effect \citep{crump2006moving}. 
In the following, we will focus on the overlap-weighted effect as it is the causal contrast that can be estimated with the lowest asymptotic variance in certain scenarios \citep{crump2006moving}. The overlap-weighted effect is defined as
  \begin{equation*}
\theta^{(1)} =  \frac{\mathbb{E}[p(T=1|X)(1-p(T=1|X)) \tau(X) ]}{\mathbb{E}[p(T=1|X)(1-p(T=1|X)) ]},
\end{equation*}
where $\tau(x)= \mathbb{E}[Y(1) - Y(0)|X=x]$. Note that if the treatment effect is homogeneous $\tau(x) \equiv \text{const.} $,  then the overlap-weighted effect and the average treatment effect coincide, that means $\theta^{(1)} = \theta^{(0)}$. Thus, shrinkage towards an efficient estimator of the overlap effect is potentially beneficial under treatment effect homogeneity. 

We  investigate shrinking between estimators of the average treatment effect and the overlap-weighted effect in a data-driven way. The proposed model selection tool will be used to trade off bias and variance.

\subsubsection{The data set} We observe $1000$ independent and identically distributed draws $(Y_{i}(T_{i}),T_{i},X_{i},S_i)$ of a distribution $\mathbb{P}$, where the $X_{i}$ are covariates.  The data generating process was chosen such that there is limited overlap, i.e.\ $\mathbb{P}[T=1|X=0] \approx 0$ and that the unconfoundedness assumptions, that means $(Y(0),Y(1)) \perp T | X$ \citep{rosenbaum1983assessing}. As discussed above, the causal effect can be estimated via doubly robust methods such as augmented inverse probability weighting, among others \citep{hernan2010causal}. The data are generated according to the following equations:
\begin{align}\label{eq:19}
  \begin{split}
    S \text{ }&\text{drawn from } \{ 1,2,3 \} \text{ uniformly at random } \\
    \epsilon_{Y} &\sim \mathcal{N}(0,1) \\
    X &\sim \text{Ber}(.5) \\
   T &\sim \begin{cases}
     \text{Ber}(.7) & \text{ if } X=1,\\
     \text{Ber}(.05) & \text{ if } X=0,
   \end{cases} \\
  Y(t) &= \frac{X}{2} + t + 3 t \gamma^{2} X  + .1t \cdot 1_{S=1} + .2t \cdot 1_{S=2} - .1t \cdot 1_{S=3}+ \epsilon_{Y},
\end{split}    
\end{align}
where $\gamma \in [0,1]$. For $\gamma = 0$, the treatment effect is homogeneous across $X$. Thus, for $\gamma=0$, the overlap-weighted effect coincides with the average treatment effect. %

\subsubsection{The estimators} In the following, we compute the estimators for each group $S=s$ separately on the data set $\{ i : S_i = s \}$. For reasons of readability, notationally we suppress the dependence of the conditional probabilities and conditional expectations on $s$. We can estimate that average treatment effect via augmented inverse probability weighting \citep{robins1994estimation},  %
\begin{equation*}
   \hat \theta_s^{(0)} = \hat \mu_{1} - \hat \mu_{0},
\end{equation*}
where 
\begin{equation*}
  \hat \mu_{a} = \frac{1}{n} \sum_{i=1}^{n}   \frac{Y_{i} 1_{T_{i} = a}}{\hat p(T_{i}=a |X_{i})} -  \frac{1_{T_{i}=a} - \hat p(T_{i}=a|X_{i})}{\hat p(T_{i} =a  | X_{i})} \hat Q(X_{i}, a),
\end{equation*}
and where $\hat Q(x,t)$ is the empirical mean of $Y$ given $X=x$ and $T=t$ and $\hat p(\cdot | \cdot)$ are empirical probabilities. %
Similarly as above, we can estimate the overlap effect by
\begin{equation*}
 \hat \theta_s^{(1)} = \frac{\hat \eta_{1} - \hat \eta_{0}}{\frac{1}{n} \sum_{i} \hat p(T_{i} =1 |X_{i}) (1 - \hat p(T_{i} = 1 | X_{i}))},
\end{equation*}
where
\begin{align*}
  \hat \eta_{a} &= \frac{1}{n} \sum_{i=1}^{n}   Y_{i} 1_{T_{i} = a} (1 - \hat p(T_{i} =a  | X_{i})) \\
  & - (1_{T_{i}=a} - \hat p(T_{i}=a|X_{i})) (1-\hat p(T_{i}=a |X_{i})) \hat Q(X_{i}, a).
\end{align*}
For $w \in \{1/10,\ldots,9/10\}$ we define
\begin{equation}
  \hat \theta(w) = (1-w) \hat \theta^{(0)} + w \hat \theta^{(1)}.
\end{equation}
For $\gamma \approx 0$, due to treatment effect homogeneity we expect $\mathbb{E}[( \hat \theta(1) - \theta^{(0)})^{2}] < \mathbb{E}[( \hat \theta^{(0)} - \theta^{(0)})^{2}]$. For $\gamma \approx 1$, we expect $\mathbb{E}[( \hat \theta(1) - \theta^{(0)})^{2}] > \mathbb{E}[( \hat \theta(0) - \theta^{(0)})^{2}]$. In the first case, the optimal estimator is $\hat \theta(w)$ with $w \approx 0$. In the second case the optimal estimator is $\hat \theta(w) $ with $w \approx 1$.

\subsubsection{Results} The mean-squared error of the estimator selected by targeted selection and 10-fold cross-validation is depicted in Figure~\ref{fig:hetero}. To further improve the stability of cross-validation, we repeat the data splitting 10 times on randomly shuffled data. To study the influence of dimension $d$ on the performance of the model selection procedure, we show how the method performs on the full data set (left-hand side) and how the method performs if it only has access to the subset of observations $i$ for which $S_i = 1$ (right-hand side). Results are averaged across $100$ simulation runs. The method performs better for $d=3$  (left plot) than for $d=1$ (right plot), which is consistent with Theorem~\ref{thm:uniform}. For $d=3$, targeted selection performs similarly or better than the baseline estimator for all $\gamma$. For $d=1$, the proposed approach performs worse than the baseline estimator for $.4 < \gamma < .9$.  %

We also evaluated the realized coverage of confidence intervals with nominal coverage $95 \%$ as described in Section~\ref{sec:infer-after-model}. Equation~\ref{eq:bootstrap} is estimated using the non-parametric bootstrap. Across $\gamma \in [0,1]$, the minimal realized coverage is $93 \%$. The maximal realized coverage is $97 \%$. Averaged across all $\gamma \in [0,1]$, the overall coverage is $94.9 \%$.

\begin{figure}
  \begin{minipage}{.49\textwidth}
    \centering
    \includegraphics[scale=.57]{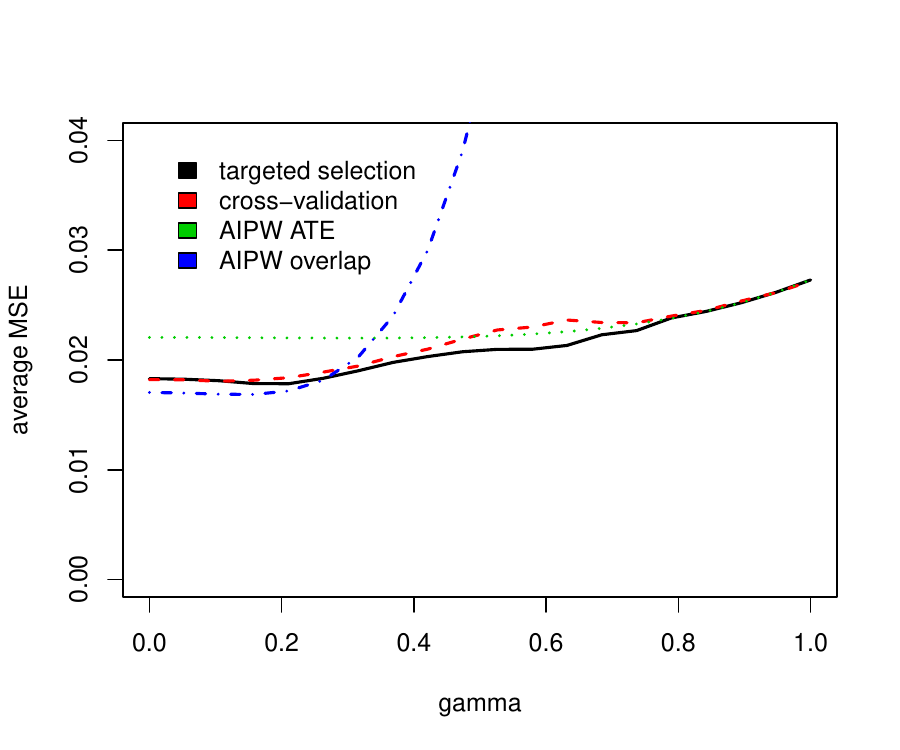}
  \end{minipage} \hfill
  \begin{minipage}{.49\textwidth}
    \centering
    \includegraphics[scale=.57]{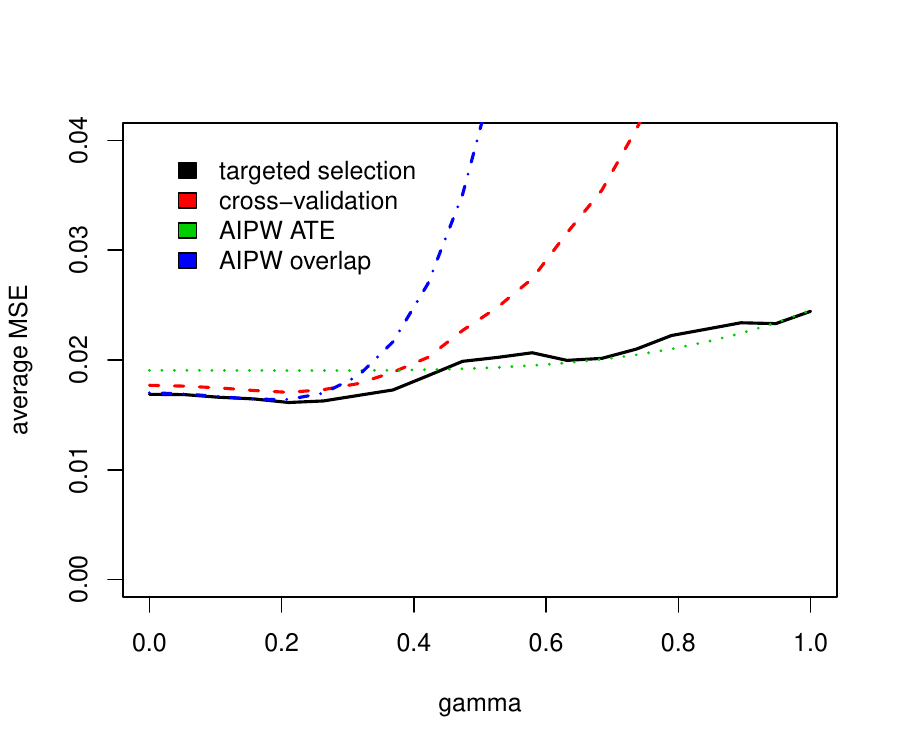}
  \end{minipage}
\begin{center}
  \caption{ Mean-squared error $R(\hat w)$ where $\hat w$ is selected via the cross-validation criterion or  targeted selection. The data are drawn according to equation~\eqref{eq:19}. Cross-validation and targeted selection are used to shrink between the AIPW ATE estimator and the AIPW overlap estimator. On the left-hand side, we show the error $  \| \hat \theta_\bullet^{(\overline g)} - \theta_\bullet^{(0)} \|_2^2/3$. On the right-hand side the method is run on the subset of observations $i$ for which $S_i=1$. Consistent with the theory presented in Section~\ref{sec:finite-sample-bound}, the maximum excess risk is smaller for $d=3$ (left) than for $d=1$ (right).  %
    The proposed method performs equal or better than cross-validation across most $\gamma \in [0,1]$.
  }\label{fig:hetero}
\end{center}
\end{figure}

\subsection{Instrumental variables and data fusion}\label{sec:stab-instr-vari-1}

 The instrumental variables approach is a widely-used method to estimate causal effect of a treatment $T$ on a target outcome $Y$ in the presence of confounding \citep{wright1928tariff,bowden1990instrumental,angrist1996identification}. Roughly speaking, the method relies on a predictor $I$ (called the instrument) of the treatment $T$ that is not associated with the error term of the outcome $Y$. We will not discuss the assumptions behind instrumental variables in detail, but refer the interested reader to \citet{hernan2010causal}. We will focus on the case, where $I$, $T$ and $Y$ are one-dimensional. Under IV assumptions and linearity, the target quantity can be re-written as
\begin{equation*}
  \theta^{(0)} = \mathbb{E}[Y(1) - Y(0)] =  \frac{\text{Cov}(I,Y)}{\text{Cov}(I,T)}.
\end{equation*}
Estimating this quantity can be challenging if the instrument is weak, i.e.\ if $\text{Cov}(I,T) \approx 0$. In this case, the approach can benefit from shrinkage towards the ordinary least-squares solution \citep{nagar1959bias,theil1961economic,rothenhausler2018anchor,jakobsen2020distributional}. Doing so may decrease the variance but generally introduces bias. We will focus on the case where we have some additional observational data, where we observe $T$ and $Y$, but where the instrument $I$ is unobserved.

\subsubsection{The data set} We draw $500$ i.i.d.\ observations according to the following equations:
\begin{align}\label{eq:30}
\begin{split}
  S \text{ }&\text{drawn from } \{ 1,2,3 \} \text{ uniformly at random } \\
  I,H,\epsilon_{T},\epsilon_{Y} &\sim \mathcal{N}(0,1)  \\
  T &= \frac{I}{2}  + H + \epsilon_{T} \\
  Y(t) &= t - \gamma^{2} H  + .1  t\cdot 1_{S=1} + .2 t \cdot 1_{S=2} - .1 t\cdot 1_{S=3} + \epsilon_{Y}
\end{split}    
\end{align}
We vary $ \gamma \in [0,2]$, which corresponds to the strength of confounding between $T$ and $Y$. %
We observe $(T_{i},Y_{i}(T_{i}),I_{i})$ for $i=1,\ldots,500$. We also assume that we have access to a larger  data set $i=501,\ldots,1000$ with incomplete observations. To be more precise, on this data set we only observe $X$ and $Y$, but not the instrument $I$. Formally, for $i = 501,\ldots,1000$ we observe $(T_{i},Y_{i}(T_{i}))$ drawn according to equation~\eqref{eq:30}. 
\subsubsection{The estimators} In the linear case, for each subset $S=s$, the instrumental variables estimator can be written as %
\begin{align*}%
  \begin{split}
  (\hat b_{IV})_s &= \frac{\hat{\text{Cov}}(I,Y|S=s)}{\hat{\text{Cov}}(I,T | S=s)}, %
\end{split}  
\end{align*}
where $\hat{\text{Cov}}$ denotes the empirical covariance over the observations $i=1,\ldots,500$.   %
To deal with the weak instrument, we will consider shrinking the instrumental variables estimator torwards ordinary least-squares,
\begin{equation*}
  (\hat b_{\text{OLS}})_s = \argmin_{b} \min_{c} \hat{\E}[(Y - T b -c)^{2}|S=s], 
\end{equation*}
where $\hat{\E}$ denotes the empirical expectation over the observations $i=1,\ldots,1000$. Shrinking towards the ordinary least-squares solution will introduce some bias if $\gamma \neq 0$, but potentially decreases variance. As candidate estimators, for any $w \in \{0/10,1/10,\ldots,10/10\}$ we consider convex combinations of OLS and IV,
\begin{equation*}
  \hat \theta(w) = w \hat b_{\text{OLS}} + (1-w)  \hat b_{\text{IV}}.
\end{equation*}

\subsubsection{Results} The mean-squared error of the estimator selected by targeted selection and cross-validation is depicted in Figure~\ref{fig:IV}. We use 10-fold cross-validation. To further improve the stability of cross-validation, we repeat the data splitting 10 times on randomly shuffled data. Results are averaged across $100$ simulation runs. Over the entire range $\gamma \in [0,2]$, the proposed approach outperforms cross-validation. For $\gamma \approx 0$ and $\gamma > 1$, targeted selection outperforms the IV approach. For $.5 < \gamma  < 1$, the proposed approach performs worse than the IV approach. We evaluate the realized coverage of confidence intervals with nominal coverage $95 \%$ as described in Section~\ref{sec:infer-after-model}. Equation~\ref{eq:bootstrap} is estimated using the non-parametric bootstrap. Across $\gamma \in [0,2]$, the minimal realized coverage was $89 \%$. The maximal realized coverage was $96 \%$.  Averaged across all $\gamma \in [0,2]$, the overall realized coverage was $93 \%$. %
 
\begin{figure}
  \centering
  \includegraphics[scale=.5]{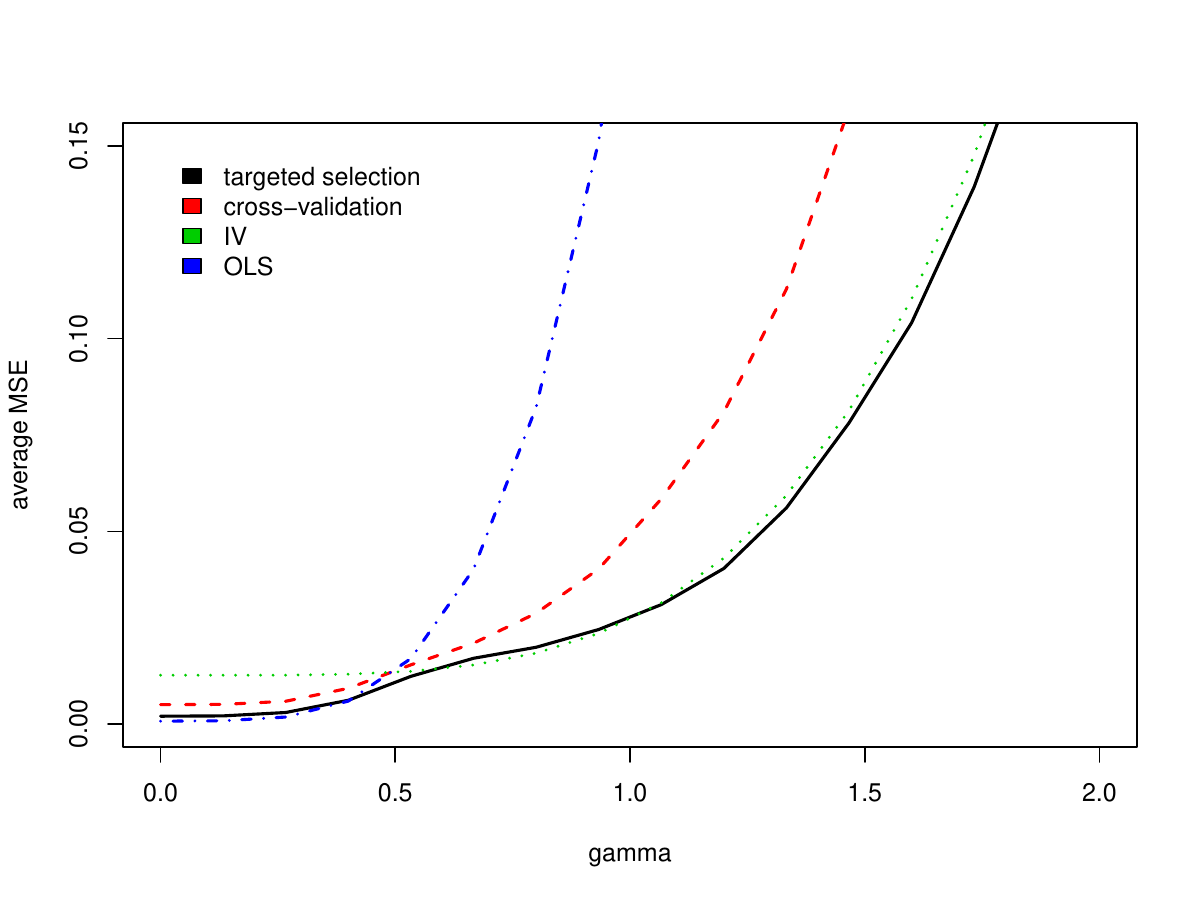}
  \caption{ Mean-squared error $R(\hat w)$ where $\hat w$ is selected via the cross-validation criterion or targeted selection. The data are drawn according to equation~\eqref{eq:30}.  Cross-validation and targeted selection is used to stabilize the instrumental variables approach by shrinking the estimate towards ordinary least-squares. The proposed method performs equal or better than cross-validation across all $\gamma \in [0,2]$. %
  }\label{fig:IV}
\end{figure}

\subsection{Experiment with proxy outcome}

One of the most popular estimators for causal effects in experimental settings is difference-in-means. To improve variance, it is possible to adjust for pre-treatment covariates, see for example \citet{lin2013agnostic}. This raises the question whether post-treatment covariates can be used to improve the precision of causal effect estimates. This is indeed the case under additional assumptions. For example, in some cases, the treatment effect can be written as the product
\begin{equation}\label{eq:12}
 \theta^{(0)} = \mathbb{E}[Y|T=1] - \mathbb{E}[Y|T=0] =  \theta_{T \rightarrow P} \cdot \theta_{P \rightarrow Y},
\end{equation}  
where  $\theta_{T \rightarrow P} =  \mathbb{E}[P|T=1] - \mathbb{E}[P|T=0]$ is the effect of the treatment on some surrogate or proxy outcome $P \in \{0,1\}$; and  $\theta_{P \rightarrow Y} =\mathbb{E}[Y|P=1] - \mathbb{E}[Y|P=0]$ is the effect of the proxy on the outcome. It is well-known that estimators that make use of such decompositions can outperform the standard difference-in-means estimator in terms of asymptotic variance \citep{tsiatis2007semiparametric,athey2019surrogate,guo2020efficient}. However, doing so can introduce bias if equation~\eqref{eq:12} does not hold. We will use the proposed model selection procedure to shrink between difference-in-means and an estimator that is unbiased if the treatment effect decomposition in equation~\eqref{eq:12} holds. %

\subsubsection{The data set} We consider a simple experimental setting with a post-treatment variable $P$. For simplicity, let us consider an experiment with binary treatment $T \in \{0,1\}$, a binary proxy outcome $P \in \{0,1\}$ and outcome $Y$.  We draw $200$ i.i.d.\ observations according to the following equations:
\begin{align}\label{eq:20}
  \begin{split} 
    S \text{ }&\text{drawn from } \{ 1,2,3 \} \text{ uniformly at random } \\
    T &\sim \text{Ber}(.5) \\
  \epsilon_{P},\epsilon_{Y} &\sim \mathcal{N}(0,1)  \\
  P(t) &= 1_{\epsilon_{p} \le t} \\
  Y(t) &=  P(t)   + \gamma^2 \left(.1t \cdot 1_{S=1} + .2t \cdot 1_{S=2} - .1t \cdot 1_{S=3} \right) + \epsilon_{Y}
\end{split}    
\end{align}
For $\gamma = 0$, the outcome $Y(T)$ is conditionally independent of the treatment, given the proxy $P(T)$. In this case, the average treatment effect can be written in product form, $\theta^{(0)} = \theta_{T \rightarrow P} \cdot \theta_{P \rightarrow Y}$, and this decomposition can be leveraged for estimation. For $\gamma \neq 0$, this decomposition does not hold.

\subsubsection{The estimators} The standard estimator to estimate causal effects from experiments is difference-in-means,
\begin{equation}\label{eq:diffinmeans}
  \hat \theta^{(0)} %
  =\frac{1}{ \sum T_{i}}\sum_{i : T_{i} = 1} Y_{i}  - \frac{1}{\sum (1-T_{i})} \sum_{i : T_{i} = 0} Y_{i}.
\end{equation}
If the proxy outcome is a valid surrogate, i.e.\ if
\begin{equation*}
   Y \perp T | P,
\end{equation*}
we can rewrite $\theta^{(0)}$ as
\begin{align*}
  \theta^{(0)} &= \mathbb{E}[Y|T=1] - \mathbb{E}[Y|T=0]  \\
  &=  \mathbb{E}[\mathbb{E}[Y|P=1] P + \mathbb{E}[Y|P=0] (1-P) |T=1] \\
         & \qquad - \mathbb{E}[\mathbb{E}[Y|P=1] P + \mathbb{E}[Y|P=0] (1-P) |T=0] \\
             &= \left(  \mathbb{E}[Y|P=1] - \mathbb{E}[Y|P=0]\right) \left( \mathbb{E}[P|T=1]  - \mathbb{E}[P|T=0]  \right)
  \\
  &= \theta_{T \rightarrow P} \cdot \theta_{P \rightarrow Y}.
\end{align*}
Thus, in this case, we can also consider the product estimator
\begin{align}\label{eq:products}
\begin{split}  
  \hat \theta^{(1)} &= \left(\frac{1}{ \sum T_{i}}\sum_{i : T_{i} = 1} P_{i}  - \frac{1}{\sum (1-T_{i})} \sum_{i : T_{i} = 0} P_{i} \right)  \\
  & \qquad \cdot \left(\frac{1}{ \sum P_{i}}\sum_{i : P_{i} = 1} Y_{i}  - \frac{1}{\sum (1-P_{i})} \sum_{i : P_{i} = 0} Y_{i} \right)
\end{split}  
\end{align}
On each subset $\{ i : S_i = s \}$ we compute \eqref{eq:diffinmeans} and \eqref{eq:products}, yielding $\hat \theta_s^{(1)}$ and $\hat \theta_s^{(0)}$ for $s = 1,2,3$. We shrink between these two vectors, i.e. for $w \in \{1/10,\ldots,9/10\}$ we define
\begin{equation*}
  \hat \theta(w) =   (1-w) \hat \theta^{(0)} + w \hat \theta^{(1)}.
\end{equation*}

\subsubsection{Results} The mean-squared error of the estimator selected by targeted selection and cross-validation is depicted in Figure~\ref{fig:proxy}. We use 10-fold cross-validation. To improve the stability of cross-validation, we repeat the data splitting 10 times on randomly shuffled data. Results are averaged across $100$ simulation runs. Similarly as above, targeted selection performs similar or better than cross-validation. Targeted selection performs better than the difference-in-means for $\gamma < 1$. For $\gamma \ge 1$, targeted selection approaches the performance of difference-in-means. We evaluated the realized coverage of confidence intervals with nominal coverage $95 \%$ as described in Section~\ref{sec:infer-after-model}. Equation~\ref{eq:bootstrap} is estimated using the non-parametric bootstrap. Across $\gamma \in [0,1.2]$, the minimal realized coverage was $93 \%$. The maximal realized coverage was $96 \%$. Averaged across all $\gamma \in [0,1]$, the overall coverage was $94.6 \%$.

\begin{figure}
\begin{center}
  \includegraphics[scale=.5]{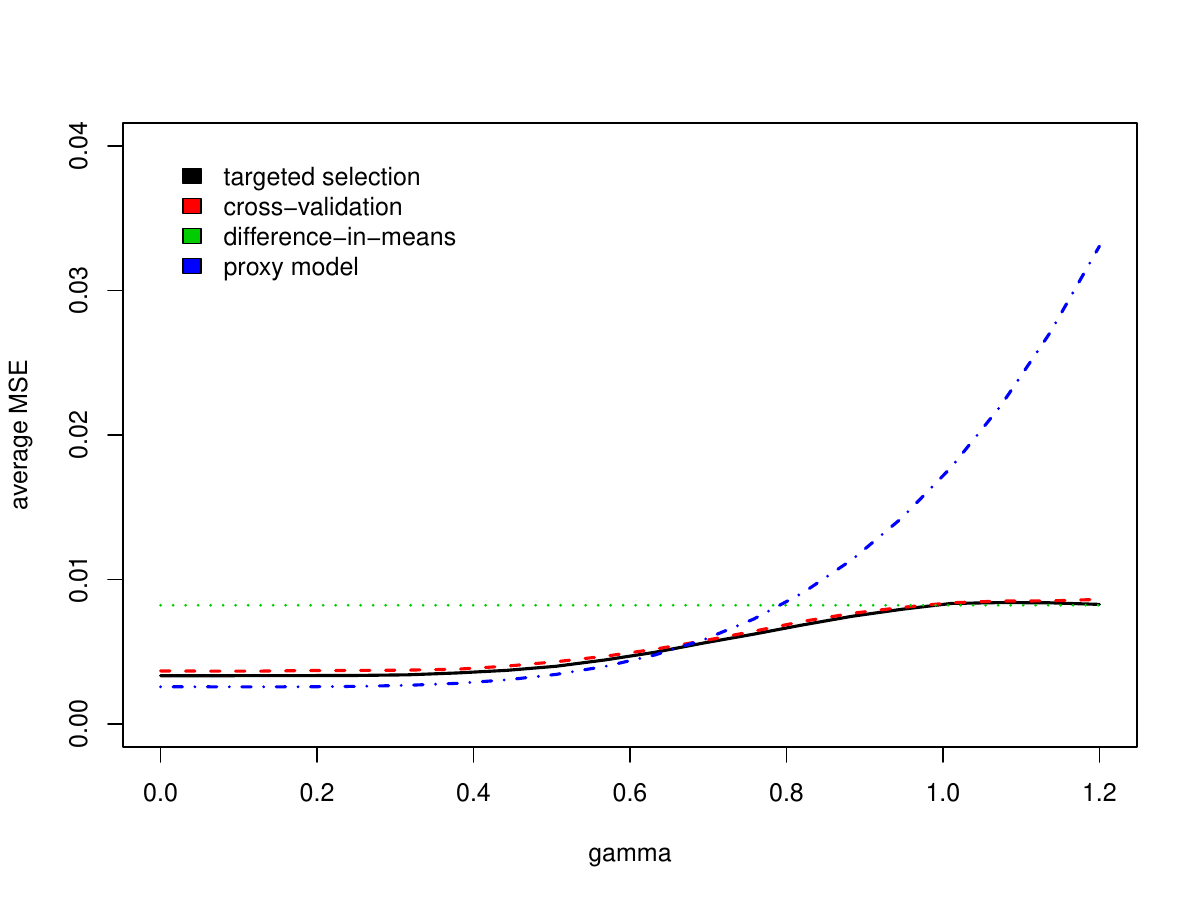}
  \caption{Mean-squared error $R(\hat w)$ where $\hat w$ is selected via the cross-validation criterion or targeted model selection. The data are drawn according to equation~\eqref{eq:20}. Cross-validation and targeted selection is used to stabilize the difference-in-means estimator by shrinking towards an estimator that makes use of a proxy outcome. Both procedures perform better than the baseline model for $\gamma \in [0,1]$. %
  }\label{fig:proxy}
\end{center}
\end{figure}

\section{Conclusion}

We have introduced a method that allows to conduct targeted parameter selection by estimating the bias and variance of candidate estimators. %
The theoretical justification of the method relies on a linear expansion of the estimator. The method is very general and can be used in both parametric and semi-parametric settings. Under regularity conditions, we showed that the proposed criterion has asymptotically equal or lower variance than competing procedures based on sample splitting. In addition, we showed that for $n \rightarrow \infty$, the modified risk criterion selects models with lower or equal risk than the baseline estimator $\hat \theta^{(0)}$. Furthermore, we derived asymptotically valid confidence intervals in low-dimensional settings. %

In simulations, we showed that the method selects reasonable models and outperforms the cross-validation procedure in most scenarios. The proposed method can decrease variance if the competing estimators are approximately unbiased. However, there is no free lunch. In transitional regimes, for fixed $n$, the estimator can perform worse than the baseline estimator $\hat \theta^{(0)}$. This is to be expected from statistical theory, see for example  \citet[page 110]{van2000asymptotic}. However, as the finite-sample bound and our simulations show, the excess risk goes to zero as the dimension of the target parameter grows. %

The theoretical justification of the proposed method relies on a linear approximation of the estimator in a neighborhood of the parameter values $\theta^{(g)}$. Thus, it would be important to understand the performance of the method in scenarios where parameter estimates of some of the estimators are far from the parameter values. %
In Section~\ref{sec:stab-instr-vari-1}, we have seen some preliminary evidence that the proposed methodology may be used to combine knowledge across data sets. The proposed method is not tailored to this special case. %
Thus, we believe that it would be exciting to investigate whether the  model selection can be further improved for  data fusion tasks.

\if1
{
\section*{Acknowledgements}

The author would like to thank Guillaume Basse, Guido Imbens, and Bin Yu for inspiring discussions.
} \fi

\newpage
\bibliographystyle{plainnat}
\bibliography{bibliography}

\begin{thebibliography}{45}
\providecommand{\natexlab}[1]{#1}
\providecommand{\url}[1]{\texttt{#1}}
\expandafter\ifx\csname urlstyle\endcsname\relax
  \providecommand{\doi}[1]{doi: #1}\else
  \providecommand{\doi}{doi: \begingroup \urlstyle{rm}\Url}\fi

\bibitem[Akaike(1974)]{akaike1974new}
H.~Akaike.
\newblock A new look at the statistical model identification.
\newblock In \emph{Selected Papers of Hirotugu Akaike}, pages 215--222.
  Springer, 1974.

\bibitem[Angrist et~al.(1996)Angrist, Imbens, and
  Rubin]{angrist1996identification}
J.~Angrist, G.~Imbens, and D.~Rubin.
\newblock Identification of causal effects using instrumental variables.
\newblock \emph{Journal of the American Statistical Association}, 91\penalty0
  (434):\penalty0 444--455, 1996.

\bibitem[Arlot and Celisse(2010)]{arlot2010survey}
S.~Arlot and A.~Celisse.
\newblock A survey of cross-validation procedures for model selection.
\newblock \emph{Statistics surveys}, 2010.

\bibitem[Athey and Imbens(2016)]{athey2016recursive}
S.~Athey and G.~Imbens.
\newblock Recursive partitioning for heterogeneous causal effects.
\newblock \emph{Proceedings of the National Academy of Sciences}, 2016.

\bibitem[Athey et~al.(2020)Athey, Chetty, and Imbens]{athey2020combining}
S.~Athey, R.~Chetty, and G.~Imbens.
\newblock Combining experimental and observational data to estimate treatment
  effects on long term outcomes.
\newblock \emph{arXiv preprint arXiv:2006.09676}, 2020.

\bibitem[Athey et~al.(2019)Athey, Chetty, Imbens, and Kang]{athey2019surrogate}
Susan Athey, Raj Chetty, Guido~W Imbens, and Hyunseung Kang.
\newblock The surrogate index: Combining short-term proxies to estimate
  long-term treatment effects more rapidly and precisely.
\newblock Technical report, National Bureau of Economic Research, 2019.

\bibitem[Bareinboim and Pearl(2016)]{bareinboim2016causal}
E.~Bareinboim and J.~Pearl.
\newblock Causal inference and the data-fusion problem.
\newblock \emph{Proceedings of the National Academy of Sciences}, 113\penalty0
  (27):\penalty0 7345--7352, 2016.

\bibitem[Berk et~al.(2013)Berk, Brown, Buja, Zhang, and Zhao]{berk2013valid}
R.~Berk, L.~Brown, A.~Buja, K.~Zhang, and L.~Zhao.
\newblock Valid post-selection inference.
\newblock \emph{The Annals of Statistics}, 41\penalty0 (2):\penalty0 802--837,
  2013.

\bibitem[Bowden and Turkington(1990)]{bowden1990instrumental}
R.~Bowden and D.~Turkington.
\newblock \emph{Instrumental variables}, volume~8.
\newblock Cambridge university press, 1990.

\bibitem[Brookhart and Van Der~Laan(2006)]{brookhart2006semiparametric}
M.~Brookhart and M.~Van Der~Laan.
\newblock A semiparametric model selection criterion with applications to the
  marginal structural model.
\newblock \emph{Computational statistics \& data analysis}, 50\penalty0
  (2):\penalty0 475--498, 2006.

\bibitem[Claeskens and Hjort(2003)]{claeskens2003focused}
G.~Claeskens and N.~Hjort.
\newblock The focused information criterion.
\newblock \emph{Journal of the American Statistical Association}, 2003.

\bibitem[Claeskens and Hjort(2008)]{claeskens2008model}
G.~Claeskens and N.~Hjort.
\newblock \emph{Model selection and model averaging}.
\newblock Cambridge University Press, 2008.

\bibitem[Crump et~al.(2006)Crump, Hotz, Imbens, and Mitnik]{crump2006moving}
R.~Crump, V.~Hotz, G.~Imbens, and O.~Mitnik.
\newblock Moving the goalposts: Addressing limited overlap in the estimation of
  average treatment effects by changing the estimand.
\newblock Technical report, National Bureau of Economic Research, 2006.

\bibitem[Cui and Tchetgen-Tchetgen(2019)]{cui2019bias}
Y.~Cui and E.~Tchetgen-Tchetgen.
\newblock Bias-aware model selection for machine learning of doubly robust
  functionals.
\newblock \emph{arXiv preprint arXiv:1911.02029}, 2019.

\bibitem[Friedman et~al.(2001)Friedman, Hastie, and
  Tibshirani]{friedman2001elements}
J.~Friedman, T.~Hastie, and R.~Tibshirani.
\newblock \emph{The elements of statistical learning}.
\newblock Springer, 2001.

\bibitem[Gill and Levit(1995)]{gill1995applications}
Richard~D Gill and Boris~Y Levit.
\newblock Applications of the van {T}rees inequality: a {B}ayesian
  {C}ram{\'e}r-{R}ao bound.
\newblock \emph{Bernoulli}, 1\penalty0 (1-2):\penalty0 59--79, 1995.

\bibitem[Guo and Perkovi{\'c}(2020)]{guo2020efficient}
F.~Guo and E.~Perkovi{\'c}.
\newblock Efficient least squares for estimating total effects under linearity
  and causal sufficiency.
\newblock \emph{arXiv preprint arXiv:2008.03481}, 2020.

\bibitem[Heckman et~al.(1998)Heckman, Ichimura, and Todd]{heckman1998matching}
J.~Heckman, H.~Ichimura, and P.~Todd.
\newblock Matching as an econometric evaluation estimator.
\newblock \emph{The review of economic studies}, 65\penalty0 (2):\penalty0
  261--294, 1998.

\bibitem[Hernan and Robins(2020)]{hernan2010causal}
M.~Hernan and J.~Robins.
\newblock \emph{Causal inference: What If}.
\newblock Chapman \& Hall, 2020.

\bibitem[Huber(1967)]{huber1967behavior}
P.~Huber.
\newblock The behavior of maximum likelihood estimates under nonstandard
  conditions.
\newblock In \emph{Proceedings of the fifth Berkeley symposium on mathematical
  statistics and probability}. University of California Press, 1967.

\bibitem[Imbens and Wooldridge(2009)]{imbens2009recent}
G.~Imbens and J.~Wooldridge.
\newblock Recent developments in the econometrics of program evaluation.
\newblock \emph{Journal of economic literature}, 47\penalty0 (1):\penalty0
  5--86, 2009.

\bibitem[Jakobsen and Peters(2020)]{jakobsen2020distributional}
M.~Jakobsen and J.~Peters.
\newblock Distributional robustness of k-class estimators and the {PULSE}.
\newblock \emph{arXiv preprint arXiv:2005.03353}, 2020.

\bibitem[Kapelner et~al.(2014)Kapelner, Bleich, Levine, C., DeRubeis, and
  Berk]{kapelner2014inference}
A.~Kapelner, J.~Bleich, A.~Levine, Zachary~D. C., R.~DeRubeis, and R.~Berk.
\newblock Inference for the effectiveness of personalized medicine with
  software.
\newblock \emph{arXiv preprint arXiv:1404.7844}, 2014.

\bibitem[LaLonde(1986)]{lalonde1986evaluating}
R.~LaLonde.
\newblock Evaluating the econometric evaluations of training programs with
  experimental data.
\newblock \emph{The American economic review}, pages 604--620, 1986.

\bibitem[Lin(2013)]{lin2013agnostic}
W.~Lin.
\newblock Agnostic notes on regression adjustments to experimental data:
  {R}eexamining {F}reedman’s critique.
\newblock \emph{The Annals of Applied Statistics}, 7\penalty0 (1):\penalty0
  295--318, 2013.

\bibitem[Nagar(1959)]{nagar1959bias}
A.~Nagar.
\newblock The bias and moment matrix of the general k-class estimators of the
  parameters in simultaneous equations.
\newblock \emph{Econometrica}, pages 575--595, 1959.

\bibitem[Nie and Wager(2017)]{nie2017quasi}
X.~Nie and S.~Wager.
\newblock Quasi-oracle estimation of heterogeneous treatment effects.
\newblock \emph{arXiv preprint arXiv:1712.04912}, 2017.

\bibitem[Peters et~al.(2016)Peters, B{\"u}hlmann, and
  Meinshausen]{peters2015causal}
Jonas Peters, Peter B{\"u}hlmann, and Nicolai Meinshausen.
\newblock Causal inference by using invariant prediction: identification and
  confidence intervals.
\newblock \emph{Journal of the Royal Statistical Society: Series B},
  78\penalty0 (5), 2016.

\bibitem[Powers et~al.(2018)Powers, Qian, Jung, Schuler, Shah, Hastie, and
  Tibshirani]{powers2018some}
S.~Powers, J.~Qian, K.~Jung, A.~Schuler, N.~Shah, T.~Hastie, and R.~Tibshirani.
\newblock Some methods for heterogeneous treatment effect estimation in high
  dimensions.
\newblock \emph{Statistics in medicine}, 2018.

\bibitem[Robins et~al.(1994)Robins, Rotnitzky, and Zhao]{robins1994estimation}
J.~Robins, A.~Rotnitzky, and L.~Zhao.
\newblock Estimation of regression coefficients when some regressors are not
  always observed.
\newblock \emph{Journal of the American Statistical Association}, 1994.

\bibitem[Rolling and Yang(2014)]{rolling2014model}
C.~Rolling and Y.~Yang.
\newblock Model selection for estimating treatment effects.
\newblock \emph{Journal of the Royal Statistical Society: Series B}, 2014.

\bibitem[Rosenbaum and Rubin(1983)]{rosenbaum1983assessing}
P.~Rosenbaum and D.~Rubin.
\newblock Assessing sensitivity to an unobserved binary covariate in an
  observational study with binary outcome.
\newblock \emph{Journal of the Royal Statistical Society: Series B},
  45\penalty0 (2):\penalty0 212--218, 1983.

\bibitem[Rothenh{\"a}usler et~al.(2020)Rothenh{\"a}usler, Meinshausen,
  B{\"u}hlmann, and Peters]{rothenhausler2018anchor}
D.~Rothenh{\"a}usler, N.~Meinshausen, P.~B{\"u}hlmann, and J.~Peters.
\newblock Anchor regression: heterogeneous data meets causality.
\newblock \emph{arXiv preprint arXiv:1801.06229}, 2020.

\bibitem[Rubin(1974)]{rubin1974estimating}
D.~Rubin.
\newblock Estimating causal effects of treatments in randomized and
  nonrandomized studies.
\newblock \emph{Journal of educational Psychology}, 66\penalty0 (5):\penalty0
  688, 1974.

\bibitem[Schuler et~al.(2018)Schuler, Baiocchi, Tibshirani, and
  Shah]{schuler2018comparison}
A.~Schuler, M.~Baiocchi, R.~Tibshirani, and N.~Shah.
\newblock A comparison of methods for model selection when estimating
  individual treatment effects.
\newblock \emph{arXiv preprint arXiv:1804.05146}, 2018.

\bibitem[Schwarz(1978)]{schwarz1978estimating}
G.~Schwarz.
\newblock Estimating the dimension of a model.
\newblock \emph{Annals of Statistics}, 6\penalty0 (2), 1978.

\bibitem[Splawa-Neyman et~al.(1990)Splawa-Neyman, Dabrowska, and
  Speed]{splawa1990application}
J.~Splawa-Neyman, D.~Dabrowska, and T.~Speed.
\newblock On the application of probability theory to agricultural experiments.
\newblock \emph{Statistical Science}, pages 465--472, 1990.

\bibitem[Taylor and Tibshirani(2015)]{taylor2015statistical}
J.~Taylor and R.~Tibshirani.
\newblock Statistical learning and selective inference.
\newblock \emph{Proceedings of the National Academy of Sciences}, 112\penalty0
  (25):\penalty0 7629--7634, 2015.

\bibitem[Theil(1961)]{theil1961economic}
H.~Theil.
\newblock \emph{Economic forecasts and policy}.
\newblock North-Holland Pub. Co., 1961.

\bibitem[Tsiatis(2007)]{tsiatis2007semiparametric}
A.~Tsiatis.
\newblock \emph{Semiparametric theory and missing data}.
\newblock Springer Science \& Business Media, 2007.

\bibitem[Van~der Laan and Robins(2003)]{van2003unified}
M.~Van~der Laan and J.~Robins.
\newblock \emph{Unified methods for censored longitudinal data and causality}.
\newblock Springer, 2003.

\bibitem[Van~der Vaart(2000)]{van2000asymptotic}
A.~Van~der Vaart.
\newblock \emph{Asymptotic statistics}.
\newblock Cambridge university press, 2000.

\bibitem[Wainwright(2019)]{wainwright2019high}
M.~Wainwright.
\newblock \emph{High-dimensional statistics: A non-asymptotic viewpoint},
  volume~48.
\newblock Cambridge University Press, 2019.

\bibitem[Wright(1928)]{wright1928tariff}
P.~Wright.
\newblock \emph{Tariff on animal and vegetable oils}.
\newblock Macmillan Company, New York, 1928.

\bibitem[Zhao et~al.(2017)Zhao, Fang, and Simchi-Levi]{zhao2017uplift}
Y.~Zhao, X.~Fang, and D.~Simchi-Levi.
\newblock Uplift modeling with multiple treatments and general response types.
\newblock In \emph{Proceedings of the 2017 SIAM International Conference on
  Data Mining}. SIAM, 2017.

\end{thebibliography}

\newpage
\pagenumbering{arabic} \setcounter{page}{1}
\appendix
\begin{appendix}
\section{SUPPLEMENTARY MATERIAL}
  
This supplementary material contains proofs for the theoretical results in the main paper.

\subsection{Proof of Theorem~\ref{theorem:unbiased}}

\begin{proof}
 First, note that
\begin{equation*}
 n ( \hat R(g) - \mathbb{E}[\| \hat \theta^{(g)} - \theta^{(0)} \|_2^2]) = \sum_{j=1}^d  n(\hat \theta_j^{(g)} - \hat \theta_j^{(0)})^2 -  \hat\tau_j^{(g)} + \hat\sigma_j^{(g)}   -  n\mathbb{E}[ ( \hat \theta_j^{(g)} - \theta_j^{(0)})^2].
\end{equation*}
Thus, it is sufficient to show that for each $j$ with $\theta_j^{(0)} = \theta_j^{(g)}$
\begin{equation*}
  n(\hat \theta_j^{(g)} - \hat \theta_j^{(0)})^2 -  \hat\tau_j^{(g)} + \hat\sigma_j^{(g)}   -  n\mathbb{E}[ ( \hat \theta_j^{(g)} - \theta_j^{(0)})^2],
\end{equation*}
converges in distribution to a centered random variable and for each $j$ with $\theta_j^{(0)} \neq \theta_j^{(g)}$,
\begin{equation*}
  \sqrt{n}(\hat \theta_j^{(g)} - \hat \theta_j^{(0)})^2 - \sqrt{n}\mathbb{E}[ ( \hat \theta_j^{(g)} - \theta_j^{(0)})^2],
\end{equation*}
converges in distribution to a centered random variable. Thus, without loss of generality, in the following we will assume that $d=1$.\\
\textbf{Case 1: $\theta^{(0)} = \theta^{(g)}$} \\
As $\mathbb{E}[ e_{g}(n)^{2}] = o(1/n)$ and as $ \mathbb{E}[\hat \theta^{(g)}] - \theta^{(g)}  = o(1/\sqrt{n})$,
  \begin{equation*}
    \mathbb{E}[ ( \hat \theta^{(g)} -  \theta^{(0)} )^{2}] = \frac{1}{n} \text{Var}(\psi^{(g)}) + o \left( \frac{1}{n}\right).
  \end{equation*}
 By assumption, %
  \begin{equation*}
     \hat \tau^{(g)} -  \tau^{(g)} = o_{P}(1).
  \end{equation*}
  Similarly, %
  \begin{equation*}
    \hat \sigma^{(g)} -  \sigma^{(g)}= o_{P}(1).
  \end{equation*}
  Combining these equations with the definition of $\hat R(g)$,
\begin{align*}
  \hat R(g) &=  (\frac{1}{n} \sum_{i=1}^{n} \psi^{(g)}(D_{i})- \psi^{(0)}(D_{i}))^{2} 
  - \frac{1}{n} \text{Var}(\psi^{(g)} - \psi^{(0)}) + \frac{1}{n} \text{Var}(\psi^{(g)})+  o_{P}(1/n).
\end{align*}
Using the Central Limit Theorem, $  \frac{1}{\sqrt{n}} \sum_{i=1}^{n} \psi^{(g)}(D_{i})- \psi^{(0)}(D_{i}) $ converges to a Gaussian random variable with mean zero and variance $\text{Var}(\psi^{(g)} - \psi^{(0)})$. Thus,
\begin{align*}
  & \qquad n \left( \hat R(g) - \mathbb{E}[(\hat \theta^{(g)} - \theta^{(0)})^{2}]\right)   \\
  &= n \left( \frac{1}{n} \left(\frac{1}{\sqrt{n}} \sum_{i=1}^{n} \psi^{(g)}(D_{i})- \psi^{(0)}(D_{i}) \right)^{2}   - \frac{1}{n} \text{Var}(\psi^{(g)} - \psi^{(0)})  \right) + o_{P}(1) \\
  &=  \left(\frac{1}{\sqrt{n}} \sum_{i=1}^{n} \psi^{(g)}(D_{i})- \psi^{(0)}(D_{i}) \right)^{2}   -  \text{Var}(\psi^{(g)} - \psi^{(0)})  + o_{P}(1)
\end{align*}
converges in distribution to a random variable with mean zero. This concludes the proof of the case $\theta^{(g)} = \theta^{(0)}$.\\
   \textbf{Case 2: $\theta^{(g)} \neq \theta^{(0)}$} \\
Similarly as above, we can show that   
\begin{align*}
  \hat R(g) &= (\theta^{(g)} - \theta^{(0)})^{2} + 2 (\theta^{(g)} - \theta^{(0)}) \frac{1}{n} \sum_{i=1}^{n} \psi^{(g)}(D_{i})- \psi^{(0)}(D_{i}) +  o_{P}(1/\sqrt{n}),
\end{align*}
and
\begin{equation*}
  \mathbb{E}[(\hat \theta^{(g)} - \theta^{(0)})^{2}] = (\theta^{(g)} - \theta^{(0)})^{2} + o(1/\sqrt{n}).
\end{equation*}
Thus,
\begin{equation}\label{eq:39}
 \sqrt{n}( \hat R(g) - \mathbb{E}[(\hat \theta^{(g)} - \theta^{(0)})^{2}])  = 2(\theta^{(g)} - \theta^{(0)})\frac{1}{\sqrt{n}}  \sum_{i=1}^{n} \psi^{(g)}(D_{i}) - \psi^{(0)}(D_{i}) + o_{P}(1)
\end{equation}
Using the CLT, $\frac{1}{\sqrt{n}}  \sum_{i=1}^{n} \psi^{(g)}(D_{i}) - \psi^{(0)}(D_{i})$ converges in distribution to a centered Gaussian random variable with variance $\text{Var}(\psi^{(g)}(D_{1}) - \psi^{(0)}(D_{1}))$. Using this fact in equation~\eqref{eq:39} concludes the proof.

\end{proof}

\subsection{Proof of Theorem~\ref{theorem:biased}}
Let $S_{k} = \{i : D_{i} \in D^{1,k}\}$.
First, note that
\begin{align*}
 &  \tilde R(g)= \sum_{j=1}^d  \xi_j + r_n,
\end{align*}
where
\begin{equation*}
  \xi_j = \frac{1}{K} \sum_{k} \left(  \theta_j^{(g)} - \theta_j^{(0)}  + \frac{1}{|S_{k}|} \sum_{i \in S_{k}} \psi_j^{(g)}(D_{i}) - \frac{1}{n-|S_{k}|} \sum_{i \not \in S_{k}} \psi_j^{(0)}(D_{i})   \right)^{2}.
\end{equation*}
and where $r_n = o_P(1/n)$ if $\theta^{(g)} = \theta^{(0)}$ and $r_n = o_P(1/\sqrt{n})$ if $\theta^{(g)} \neq \theta^{(0)}$.
Furthermore, using Assumption~\ref{assumptions},
\begin{equation*}
  \mathbb{E}[\| \hat \theta^{(g)} - \theta^{(0)} \|_2^2] = \sum_{j=1}^d (\theta_j^{(g)} - \theta_j^{(0)})^2 + \frac{1}{n} \text{Var}(\psi_j^{(g)}) + o(1/n)
\end{equation*}
Thus, it is sufficient to show the following: For each $j$ with $\theta_j^{(0)} = \theta_j^{(g)}$
\begin{equation*}
  n \left(\xi_j  - \frac{K}{n(K-1)} \text{Var}(\psi_j^{(g)})  - \frac{K}{n} \text{Var}(\psi_j^{(0)}) \right)
\end{equation*}
converges in distribution to a centered random variable and for each $j$ with $\theta_j^{(0)} \neq \theta_j^{(g)}$,
\begin{equation*}
 \sqrt{n} \left(  \xi_j   - (\theta_j^{(g)} - \theta_j^{(0)})^2 - \frac{1}{n} \text{Var}(\psi_j^{(g)}) \right)
\end{equation*}
converges in distribution to a centered random variable. In the following, without loss of generality we will assume that $d=1$.\\

\begin{proof}
\textbf{Case 1: $\theta^{(g)} = \theta^{(0)}$.} 
Recall that $|S_{k}| \sim n (K-1)/K$.  Using the CLT, for every $k=1,\ldots,K$,
\begin{equation*}
  \sqrt{n} \left( \frac{1}{|S_{k}|} \sum_{i \in S_{k}} \psi^{(g)}(D_{i}) - \frac{1}{n-|S_{k}|} \sum_{i \not \in S_{k}} \psi^{(0)}(D_{i})  \right)
\end{equation*}
converges to a centered Gaussian random variable with variance
\begin{equation*}
 \frac{1}{1-\alpha}  \text{Var}(\psi^{(g)}(D_{1})) + \frac{1}{\alpha} \text{Var}(\psi^{(0)}(D_{1})),
\end{equation*}
where $\alpha = \frac{1}{K}$.
 Hence,
\begin{equation*}
 n \xi = \frac{n}{K} \sum_{k} \left( \frac{1}{|S_{k}|} \sum_{i \in S_{k}} \psi^{(g)}(D_{i}) - \frac{1}{n-|S_{k}|} \sum_{i \not \in S_{k}} \psi^{(0)}(D_{i})   \right)^{2} 
\end{equation*}
converges to a random variable with asymptotic mean
\begin{align*}
     & \frac{1}{(1-\alpha)} \text{Var}(\psi^{(g)}(D_{1})) + \frac{1}{\alpha} \text{Var}(\psi^{(0)}(D_{1})) \\
     &=  \frac{K}{(K-1)} \text{Var}(\psi^{(g)}(D_{1})) + K \text{Var}(\psi^{(0)}(D_{1}))
   \end{align*}
   This concludes the proof of case 1.\\
   \textbf{Case 2: $\theta^{(g)} \neq \theta^{(0)}$} \\
Similarly as above, we can show that   
\begin{align*}
   \xi &= ( \theta^{(g)} - \theta^{(0)})^{2} + 2 ( \theta^{(g)} - \theta^{(0)}) \frac{1}{n} \sum_{i=1}^{n} \psi^{(g)}(D_{i})- \psi^{(0)}(D_{i}) +  o_{P}(1/\sqrt{n}).
\end{align*}
Thus,
\begin{equation}\label{eq:37}
  \sqrt{n} \left(  \xi   - (\theta^{(g)} - \theta^{(0)})^2 - \frac{1}{n} \text{Var}(\psi^{(g)}) \right) = 2(\theta^{(g)} - \theta^{(0)})\frac{1}{\sqrt{n}}  \sum_{i=1}^{n} \psi^{(g)}(D_{i}) - \psi^{(0)}(D_{i}) + o_{P}(1).
\end{equation}
Using the CLT, $\frac{1}{\sqrt{n}}  \sum_{i=1}^{n} \psi^{(g)}(D_{i}) - \psi^{(0)}(D_{i})$ converges to a centered Gaussian random variable with variance $\text{Var}(\psi^{(g)}(D_{1}) - \psi^{(0)}(D_{1}))$. Using this fact in  equation~\eqref{eq:37} concludes the proof.
\end{proof}

\subsection{Proof of Theorem~\ref{theorem:variance}}
\begin{proof}

First, we show that it is sufficient to show the statement in the one-dimensional case, i.e.\ for $d=1$. Define $\hat \eta^{(g)} = O \hat \theta^{(g)}$ and $\eta^{(g)} = O \theta^{(g)}$, where the orthonormal matrix $O \in \mathbb{R}^{d \times d}$ is chosen such that the asymptotic covariance of $\sqrt{n} (\hat \eta^{(g)} - \hat \eta^{(0)} -  \eta^{(g)}  + \eta^{(0)})$ is diagonal. By Assumption~\ref{assumptions}, $\sqrt{n} (\hat \eta^{(g)} -  \eta^{(g)} )$ converges to a Gaussian random variable with a covariance matrix that is positive definite. Let $\Sigma^{(g)}$ denote the asymptotic covariance of $\sqrt{n}(\hat \theta^{(g)} - \theta^{(g)})$. Define $\tau_j^{(g,\eta)}$ as the asymptotic standard deviation of $\sqrt{n}( \hat \eta_j^{(g)} - \hat \eta_j^{(0)} - \eta_j^{(g)} + \eta_j^{(0)} )  $ and $\sigma_j^{(g,\eta)}$ as the asymptotic standard deviation of $\sqrt{n}(\hat \eta_j^{(g)} - \eta_j^{(g)})$. Similarly let $\Sigma^{(g,\eta)}$ denote the asymptotic covariance of $\sqrt{n}( \hat \eta^{(g)} - \eta^{(g)}) $. Now note that
\begin{align*}
  &\sum_{j=1}^d ( \sigma_j^{(g)})^2 \\
  &= \text{Trace}(\Sigma^{(g)})  \\
  &= \text{Trace}(\Sigma^{(g)} O^\intercal O)  \\
  &= \text{Trace}( O \Sigma^{(g)} O^\intercal)  \\
  &= \text{Trace}(\Sigma^{(g,\eta)})  \\
  &= \sum_{j=1}^d ( \sigma_j^{(g,\eta)})^2.
\end{align*}
Analogously,
\begin{equation*}
  \sum_{j=1}^d ( \tau_j^{(g)})^2= \sum_{j=1}^d ( \tau_j^{(g,\eta)})^2.
\end{equation*}
Thus,
\begin{align*}
  \hat R(g) &= \| \hat \theta^{(g)} - \hat \theta^{(0)} \|_2^2  - \frac{1}{n}\| \hat \tau^{(g)} \|_2^2 +  \frac{1}{n} \| \hat \sigma^{(g)} \|_2^2 \\
  &= \| \hat \theta^{(g)} - \hat \theta^{(0)} \|_2^2  -  \frac{1}{n}\| \tau^{(g)} \|_2^2  +  \frac{1}{n} \| \sigma^{(g)} \|_2^2 + o_P(1/n)\\
  &= \| \hat \eta^{(g)} - \hat \eta^{(0)} \|_2^2  - \frac{1}{n} \| \tau^{(g,\eta)} \|_2^2 + \frac{1}{n}\| \sigma^{(g,\eta)} \|_2^2 + o_P(1/n).
\end{align*}
Now, by construction $\sqrt{n}(\hat \eta^{(g)} - \hat \eta^{(0)} - \eta^{(g)} + \eta^{(0)})$ converges to a Gaussian with diagonal covariance matrix. Thus, the components of $\sqrt{n}(\hat \eta^{(g)} - \hat \eta^{(0)} - \eta^{(g)} + \eta^{(0)})$ are asymptotically independent. Thus, the asymptotic mean and variance of $ \hat R(g)$
is equal to the sums of asymptotic means and asymptotic variances of
\begin{equation*}
  \hat R_j(g) := ( \hat \eta_j^{(g)} - \hat \eta_j^{(0)} )^2  - \frac{1}{n}  (\tau_j^{(g,\eta)} )^2 + \frac{1}{n}(\sigma_j^{(g,\eta)} )^2.
\end{equation*}
Similarly, note that
\begin{align*}
  \tilde R(g) &=   \frac{1}{K} \sum_{k=1}^{K} \| \hat \theta^{(g)} (D^{1,k}) - \hat \theta^{(0)}(D^{0,k}) \|_2^{2} \\
  &= \frac{1}{K} \sum_{k=1}^{K} \| \hat \eta^{(g)} (D^{1,k}) - \hat \eta^{(0)}(D^{0,k}) \|_2^{2}.
\end{align*}
As the components of $\sqrt{n}(\hat \eta^{(g)} - \hat \eta^{(0)} - \eta^{(g)} + \eta^{(0)})$ are asymptotically independent,
the asymptotic mean and variance of $\tilde R(g)$
is equal to the sums of asymptotic means and asymptotic variances of
\begin{equation*}
 \tilde R_j := \frac{1}{K} \sum_{k=1}^{K} ( \hat \eta_j^{(g)} (D^{1,k}) - \hat \eta_j^{(0)}(D^{0,k}) )^{2}. 
\end{equation*}
Thus, in the following without loss of generality we will focus on the case $d=1$ and $\hat \eta^{(g)} = \hat \theta^{(g)}$, $\eta^{(g)} = \theta^{(g)}$. We will now consider two cases. \\
  \textbf{Case 1: $\theta^{(g)} = \theta^{(0)}$} \\
Let $S_{k} = \{i : D_{i} \in D^{1,k}\}$.  Inspecting the proof of Theorem~\ref{theorem:unbiased} we obtain that
\begin{align*}
  & \qquad n \left( \hat R(g) - \mathbb{E}[(\hat \theta^{(g)} - \theta^{(0)})^{2}]\right)   \\
  &=  \left(\frac{1}{\sqrt{n}} \sum_{i=1}^{n} \psi^{(g)}(D_{i})- \psi^{(0)}(D_{i}) \right)^{2}   -  \text{Var}(\psi^{(g)} - \psi^{(0)})  + o_{P}(1).
\end{align*}
Multiplying with $\alpha = 1/K$, we can rewrite this as
\begin{align*}
  & \qquad \frac{n}{K} \left( \hat R(g) - \mathbb{E}[(\hat \theta^{(g)} - \theta^{(0)})^{2}]\right)   \\
  &=  \left(\frac{1}{ \sqrt{K} \sqrt{n}}  \sum_{i=1}^{n} \psi^{(g)}(D_{i})- \psi^{(0)}(D_{i}) \right)^{2}   - \frac{1}{K} \text{Var}(\psi^{(g)} - \psi^{(0)})  + o_{P}(1).
\end{align*}
Writing $Z_{k} = \frac{1}{\sqrt{ \alpha n }} \sum_{i \not \in S_{k}} \psi^{(g)}(D_{i})$ and $X_{k} = \frac{1}{\sqrt{ \alpha n}} \sum_{i \not \in S_{k}} \psi^{(0)}(D_{i})$ and using $ \alpha = \frac{1}{K}$ we obtain
\begin{align*}
  & \qquad \frac{n}{K} \left( \hat R(g) - \mathbb{E}[(\hat \theta^{(g)} - \theta^{(0)})^{2}]\right)   \\
  &=  \left(\frac{1}{K} \sum_{k=1}^{K} Z_{k} - X_{k} \right)^{2}   - \frac{1}{K} \text{Var}(\psi^{(g)} - \psi^{(0)})  + o_{P}(1).
\end{align*}
Similarly, by inspecting the proof of Theorem~\ref{theorem:biased},
\begin{align*}
 & n (\tilde R(g) - \mathbb{E}[(\hat \theta^{(g)} - \theta^{(0)})^{2}] ) = \\
& n  \left( \frac{1}{K} \sum_{k} \left( \frac{1}{|S_{k}|} \sum_{i \in S_{k}} \psi^{(g)}(D_{i}) - \frac{1}{n-|S_{k}|} \sum_{i \not \in S_{k}} \psi^{(0)}(D_{i})   \right)^{2} - \frac{1}{n} \text{Var}(\psi^{(g)})  \right) + o_{P}(1).
\end{align*}
We have $ (n- |S_{k} |) \sim  \alpha n$, and $ |S_{k}| \sim (K-1) \alpha n $. Thus, multiplying with $ \alpha = 1/K$ we can rewrite this as
\begin{align*}
 & \frac{n}{K} (\tilde R(g) - \mathbb{E}[(\hat \theta^{(g)} - \theta^{(0)})^{2}] ) = \\
  & \alpha n  \Bigg( \frac{1}{K } \sum_{k} \left( \frac{1}{K-1} \frac{1}{ \alpha n} \sum_{i \in S_{k}} \psi^{(g)}(D_{i}) - \frac{1}{ \alpha n} \sum_{i \not \in S_{k}} \psi^{(0)}(D_{i})   \right)^{2} \\
    &- \frac{1}{n} \text{Var}(\psi^{(g)})  \Bigg) + o_{P}(1). \\
  &= \frac{1}{K} \sum_{k} \left( \frac{1}{K-1} \frac{1}{\sqrt{ \alpha n} } \sum_{i \in S_{k}} \psi^{(g)}(D_{i}) - \frac{1}{\sqrt{ \alpha n}} \sum_{i \not \in S_{k}} \psi^{(0)}(D_{i})   \right)^{2} \\
   & -  \frac{1}{K} \text{Var}(\psi^{(g)})    + o_{P}(1) \\
  &= \frac{1}{K} \sum_{k} \left( \frac{1}{K-1}  \sum_{j \neq k}   Z_{j} -  X_{k}  \right)^{2} -  \frac{1}{K} \text{Var}(\psi^{(g)})   + o_{P}(1).
\end{align*}
 Thus, to complete the proof, we have to show that the asymptotic variance of the first term is larger than the asymptotic variance of the second term:
\begin{align*}
  &\frac{n}{K} (\tilde R(g) - \mathbb{E}[(\hat \theta^{(g)} - \theta^{(0)})^{2}] ) \\
  &= \frac{1}{K} \sum_{k} \left( \frac{1}{K-1}  \sum_{j \neq k}   Z_{j} -  X_{k}  \right)^{2}  -\frac{1}{K} \text{Var}(\psi^{(g)})   + o_{P}(1) \\
  &\frac{n}{K} \left( \hat R(g) - \mathbb{E}[(\hat \theta^{(g)} - \theta^{(0)})^{2}]\right) \\
  &=  \left(\frac{1}{K} \sum_{k=1}^{K} Z_{k} - X_{k} \right)^{2}   - \frac{1}{K} \text{Var}(\psi^{(g)} - \psi^{(0)})  + o_{P}(1).
\end{align*}
Using the CLT, for $n \rightarrow \infty$ and $K$ fixed, $(X_{1},\ldots,X_{K},Z_{1},\ldots,Z_{K})$ converge to a centered multivariate Gaussian vector with non-degenerate variance.  Hence, without loss of generality in the following we can assume that  the vector $$(X_{1},\ldots,X_{K},Z_{1},\ldots,Z_{K})$$ is multivariate Gaussian. Recall that from Assumption~\ref{assumptions} it follows that  $|\text{Cor}(\psi^{(g)},\psi^{(0)}) | \neq 1$. Thus we also have $|\text{Cor}(Z_{j},X_{j})| \neq 1$. Using Lemma~\ref{lemma:gauss2} completes the proof of case 1.

    \textbf{Case 2: $\theta^{(g)} \neq \theta^{(0)}$}

    Inspecting the proofs of Theorem~\ref{theorem:unbiased} and Theorem~\ref{theorem:biased}, we obtain that in the case $\theta^{(g)} = \theta^{(0)}$, the asymptotic variance of $\sqrt{n} ( \hat R(g) - \mathbb{E}[(\hat \theta^{(g)} - \theta^{(0)})^{2}]$ is
\begin{equation}
4 (\theta^{(g)} - \theta^{(0)})^{2}  \text{Var}(\psi^{(g)}(D_{1}) - \psi^{(0)}(D_{1})),
\end{equation}
which is the same as the asymptotic variance of $\sqrt{n} ( \tilde  R(g) - \mathbb{E}[(\hat \theta^{(g)} - \theta^{(0)})^{2}])$. This concludes the proof.
\end{proof}
\subsection{Proof of Corollary~\ref{cor:opt}} 
\begin{proof}
  By assumption, we have $\hat \theta^{(g)} - \hat \theta^{(0)} = \theta^{(g)} - \theta^{(0)} + O_{P}(1/\sqrt{n})$. If $\theta^{(g)} - \theta^{(0)} \not \equiv (0,\ldots,0)$, then
\begin{equation*}
  \hat R^{\text{mod}}(g) = \|\theta^{(g)} - \theta^{(0)} \|_2^{2} + O_{P}(1/\sqrt{n}),
\end{equation*}
with $\|\theta^{(g)} - \theta^{(0)} \|_2^{2}  > 0$. On the other hand, if $\theta^{(g)} = \theta^{(0)}$,
\begin{equation*}
   \hat R^{\text{mod}}(g) =  O_{P}(1/n).
 \end{equation*}
 Thus,
\begin{equation*}
     \mathbb{P}[ \theta^{(\bar g)} = \theta^{(0)}] \rightarrow 1.
   \end{equation*}
   Now consider any $g$ with $\theta^{(g)} = \theta^{(0)}$ and $\sum_j \text{Var}(\psi_j^{(g)}) > \sum_j \text{Var}(\psi_j^{(0)})$. Then, using Assumption~\ref{assumptions},
\begin{equation*}
   \hat R^{\text{mod}}(g)  \ge \sum_j \frac{1}{n} \text{Var}(\psi_j^{(g)}) + o_{P}(1/n),
\end{equation*}   
and
\begin{equation*}
  \hat R^{\text{mod}}(0) = \sum_j \frac{1}{n} \text{Var}(\psi_j^{(0)}) ) + o_{P}(1/n).
\end{equation*}
 Recall that by assumption $\sum_j \text{Var}(\psi_j^{(g)}) > \sum_j\text{Var}(\psi_j^{(0)})$.
Thus, $ \mathbb{P}[  \hat R^{\text{mod}}(0) <  \hat R^{\text{mod}}(g) ] \rightarrow 1$ for $n \rightarrow \infty$. As this holds for all $g$ with $\sum_j \text{Var}(\psi_j^{(g)}) > \sum_j \text{Var}(\psi_j^{(0)})$ for $n \rightarrow \infty$, this concludes the proof.
\end{proof} 
\subsection{Proof of Theorem~\ref{thm:ci}}
\begin{proof}
  We will first prove the first statement. Note that if $\hat \Sigma$ is positive definite, $\beta \mapsto b_{D_1,\ldots,D_n}(\beta)$ is continuous and strictly increasing  and $\lim_{\beta \rightarrow \infty} b_{D_1,\ldots,D_n}(\beta) = 1$, $\lim_{\beta \rightarrow - \infty} b_{D_1,\ldots,D_n}(\beta) = 0$. As $\hat \Sigma$ converges in probability to a positive definite matrix, the inverse $b^{-1}_{D_1,\ldots,D_n} : (0,1) \rightarrow \mathbb{R}$ is well-defined, except on an event with vanishing probability as $n \rightarrow \infty$.

  Let us now turn to the second statement.
  We use the decomposition 
\begin{align}\label{eq:add-decomp}
  \mathbb{P}[\sqrt{n}(\hat \theta_j^{(\bar g)} - \theta_j^{(0)}) \le \beta] = \sum_g \mathbb{P}[ \sqrt{n}(\hat \theta_j^{(g)}  - \theta_j^{(0)} )\le \beta; \bar g =g].
\end{align}
Define the event
\begin{align*}
  &A_g = \{ \bar g = g \}  \\
  &= \{\max(n \| \hat \theta^{(g)} - \hat \theta^{(0)} \|_2^2 -   \|\hat \tau^{(g)} \|_2^2,0) +  \|\hat \sigma^{(g)}\|_2^2 < \min_{g' \neq g} \max(n \|\hat \theta^{(g')} - \hat \theta^{(0)} \|_2^2 -  \| \hat \tau^{(g')} \|_2^2,0) + \|\hat \sigma^{(g')} \|_2^2 \}.
\end{align*}
 It is crucial to understand how this event behaves in the limit. $\sqrt{n}(\hat \theta - \theta)$ converges to a centered Gaussian distribution with covariance matrix $\Sigma$. Thus, define
\begin{equation*}
  A_g^{lim} = \max( \| Z_g^0 - Z_0^0 \|_2^2 -  \| \tau^{(g)} \|_2^2,0) + \|\sigma^{(g)} \|_2^2 < \min_{g' \neq g; \theta^{(g')} = \theta^{(0)}} \max( \| Z_{g'}^0 - Z_0^0 \|_2^2 -  \| \tau^{(g')} \|_2^2,0) +  \| \sigma^{(g)'} \|_2^2,
\end{equation*}
where $(Z_0^0,\ldots,Z_G^0) \sim \mathcal{N}(0,\Sigma)$, independent of the data $\{ D_j \}_{j \in \mathbb{N}}$. Before we investigate the large-sample behaviour of equation~\eqref{eq:add-decomp}, let us fix some notation.
\begin{align*}
   f_g(\beta) &:= \mathbb{P}[ \{\sqrt{n}(\hat \theta_j^{(g)} - \theta_j^{(0)}) \le \beta \} \cap A_g ] \\
  f_g^{lim}(\beta) &:= \begin{cases}\mathbb{P}[ \{ (Z_g^0)_j  \le \beta \} \cap  A_g^{lim}] & \text{ if } \theta^{(g)} = \theta^{(0)}, \\
    0 & \text{ else}.
  \end{cases} \\
  f_g^{comp}(\beta,D_1,\ldots,D_n) &:=    \mathbb{P}[  Z_{g,j} -  \sqrt{n} \hat \theta_j^{(0)}\le \beta; \max( \| Z_g - Z_0 \|_2^2 -  \| \hat \tau^{(g)} \|_2^2,0) +  \| \hat \sigma^{(g)} \|_2^2 \\
  &  \qquad  < \min_{g' \neq g} \max( \| Z_g - Z_0 \|_2^2 -  \| \hat \tau^{(g)} \|_2^2,0) +  \| \hat \sigma^{(g)} \|_2^2| D_1,\ldots,D_n],
\end{align*}
where $(Z_0,\ldots,Z_G) \sim \mathcal{N}(\sqrt{n}\hat \theta(D), \hat \Sigma(D))$, conditionally on the data $D$.  In the first step, let us consider $g$ for which $\theta^{(g)} = \theta^{(0)}$. Recall that $\hat \sigma^{(g)} \rightarrow \sigma^{(g)}$, $\hat \tau^{(g)} \rightarrow \tau^{(g)}$, and $\hat \Sigma \rightarrow \Sigma$. Also, $\sqrt{n}(\hat \theta - \theta)$ converges weakly to a distribution that is equal to the distribution of $Z^0$. %
  By weak convergence, for all $\beta$,
\begin{equation}\label{eq:lim-unbiased}
  \lim_n f_g(\beta) = f_g^{lim}(\beta) = \lim_n f_g^{comp}(\beta,D_1,\ldots,D_n),
\end{equation}
where the limit on the right-hand side is in probability. Now, let us focus on $g$ for which $\theta^{(g)} \neq \theta^{(0)}$. Note that for all $g$ with $\theta^{(g)} \neq \theta^{(0)}$, $n  \| \hat \theta^{(g)} - \hat \theta^{(0)} \|_2^2 \rightarrow \infty$. Thus, for all $g $ with $\theta^{(g)} \neq \theta^{(0)}$,
\begin{align*}
  &\mathbb{P}[\sqrt{n}(\hat \theta_j^{(g)}  - \theta_j^{(0)})\le \beta; \bar g =g] \\
  &  \le \mathbb{P}[   \max(n \| \hat \theta^{(g)} - \hat \theta^{(0)} \|_2^2 -  \| \hat \tau^{(g)} \|_2^2,0) +   \| \hat \sigma^{(g)} \|_2^2 < \min_{g' \neq g} \max(n \| \hat \theta^{(g)} - \hat \theta^{(0)} \|_2^2 -  \| \hat \tau^{g} \|_2^2,0) +  \| \hat \sigma^{(g)} \|_2^2] \\
    & \le \mathbb{P}[   \max(n \|\hat \theta^{(g)} - \hat \theta^{(0)} \|_2^2 -  \| \hat \tau^{(g)} \|_2^2,0) +   \| \hat \sigma^{(g)} \|_2^2 <    \|\hat \sigma^{(0)} \|_2^2].
  \end{align*}
Since $\hat \sigma^{(0)} \rightarrow \sigma^{(0)}$ and ${\hat \sigma^{(g)}} \rightarrow \sigma^{(g)}$ and $\hat \tau^{(g)} \rightarrow \tau^{(g)}$ and $n  \| \hat \theta^{(g)} - \hat \theta^{(0)} \|_2^2 \rightarrow \infty$,
for all $g $ with $\theta^{(g)} \neq \theta^{(0)}$ we have $\mathbb{P}[\sqrt{n}(\hat \theta_j^{(g)}  - \theta_j^{(0)})\le \beta; \bar g =g] \rightarrow 0$. Thus for all $g$ with $\theta^{(g)} \neq \theta^{(0)}$, $f_g(\beta) \rightarrow 0$. Analogously it can be shown that $ f_g^{comp}(\beta,D) \rightarrow 0$ for all $g$ with $\theta^{(g)} \neq \theta^{(0)}$. Thus, for $g$ with $\theta^{(g)} \neq \theta^{(0)}$,
\begin{equation}\label{eq:lim}
  \lim_n f_g(\beta) = f_g^{lim}(\beta) = \lim_n f_g^{comp}(\beta),
\end{equation}
where the limit on the right-hand side is in probability. By equation~\eqref{eq:lim-unbiased} and equation~\eqref{eq:lim}, for all $g$,
\begin{equation*}
  \lim_n f_g(\beta) =  f_g^{lim}(\beta) = \lim_n f_g^{comp}(\beta,D_1,\ldots,D_n),
\end{equation*}
Now note that as $\Sigma$ is positive definite, $(Z_0^0,\ldots,Z_G^0)$ has positive density on $\mathbb{R}^{d G}$ and thus
\begin{equation*}
  \beta \mapsto  f_0^{lim}(\beta)
\end{equation*}
is strictly increasing. By definition, for all $g=1,\ldots,G$, $\beta \mapsto f_g^{lim}(\beta)$ is non-decreasing. Thus, the function 
\begin{equation*}
 \beta \mapsto \sum_g f_g^{lim}(\beta)
\end{equation*}
is strictly increasing. Note that by construction $\{ A_g^{lim } : \theta^{(g)} = \theta^{(0)} \}$ form a disjoint partition of the sample space. Thus, using the definition of $f_g^{lim}$,
\begin{equation*}
\lim_{\beta \rightarrow \infty} \sum_g f_g^{lim}(\beta) = 1 \text{ and }
\lim_{\beta \rightarrow - \infty} \sum_g f_g^{lim}(\beta) = 0.
\end{equation*}
Similarly, by definition, $\beta \mapsto \sum_g f_g(\beta)$ and $\beta \mapsto \sum_g f_g^{comp}(\beta)$ are increasing with $ 0 \le  \sum_g f_g(\beta) \le 1$ and $0 \le \sum_g f_g^{comp}(\beta) \le 1$. Invoking Polya's theorem with equation~\eqref{eq:lim}, $\sup_\beta |f_g(\beta) - f_g^{lim}(\beta)| \rightarrow 0$ in probability. Analogously, $ \sup_{\beta} |f_g^{comp}(\beta,D_1,\ldots,D_n) - f_g^{lim}(\beta)|$ converges to zero in probability.

Consider a sequence $n \mapsto c_n := b^{-1}_{D_1,\ldots,D_n}(1-\alpha/2)$ such that $$\sum_g f_g^{comp}( c_n ,D_1,\ldots,D_n) = 1-\alpha/2.$$ Since $ f_g^{comp}(\beta)$ converges to $f_g^{lim}(\beta)$ uniformly, we have that $\sum_g f_g^{lim}(c_n) \rightarrow 1-\alpha/2$ in probability. Since $\beta \mapsto \sum_g f_g^{lim}(\beta)$ is strictly increasing and continuous, $c_n$ converges in probability to the unique $c^0 \in \mathbb{R}$ with $\sum_g f_g^{lim}(c^0) = 1-\alpha/2$.
Thus, for all $\epsilon > 0$,
\begin{align*}
  &\mathbb{P}[ \sqrt{n} (\hat \theta_j^{(\bar g)} - \theta_{j}^{(0)} ) \le b^{-1}_{D_1,\ldots,D_n}(1-\alpha/2)] \\
  &= \mathbb{P}[ \sqrt{n} (\hat \theta_j^{(\bar g)} - \theta_{j}^{(0)} ) \le c_n ] \\
  &\le \mathbb{P}[ \sqrt{n} (\hat \theta_j^{(\bar g)} - \theta_{j}^{(0)} ) \le c^0 + \epsilon ] + \mathbb{P}[ |c_n - c^0| \ge \epsilon ] \\
  &= \sum_g f_g(c^0+ \epsilon) + o(1).
\end{align*}
Now, letting $\epsilon \rightarrow 0$ and using that $\beta \mapsto f_g^{lim}(\beta)$ is continuous and $\sup_\beta | f_g(\beta) - f_g^{lim}(\beta) | \rightarrow 0$,
\begin{equation}\label{eq:limsup}
  \lim \sup_{n \rightarrow \infty} \mathbb{P}[ \sqrt{n} (\hat \theta_j^{(\bar g)} - \theta_{j}^{(0)} ) \le b^{-1}_{D_1,\ldots,D_n}(1-\alpha/2)] \le  \sum_g f_g^{lim}(c^0) = 1- \alpha/2.
\end{equation}
Analogously,
\begin{align*}
  &\mathbb{P}[ \sqrt{n} (\hat \theta_j^{(\bar g)} - \theta_{j}^{(0)} ) \le b^{-1}_{D_1,\ldots,D_n}(1-\alpha/2)] \\
  &= \mathbb{P}[ \sqrt{n} (\hat \theta_j^{(\bar g)} - \theta_{j}^{(0)} ) \le c_n ] \\
  &\ge \mathbb{P}[ \sqrt{n} (\hat \theta_j^{(\bar g)} - \theta_{j}^{(0)} ) \le c^0 - \epsilon ] + \mathbb{P}[ |c_n - c^0| \ge \epsilon ] \\
  &= \sum_g f_g(c^0- \epsilon) + o(1).
\end{align*}
Thus,
\begin{align}\label{eq:liminf}
  \lim \inf_{n \rightarrow \infty} \mathbb{P}[ \sqrt{n} (\hat \theta_j^{(\bar g)} - \theta_{j}^{(0)} ) \le b^{-1}_{D_1,\ldots,D_n}(1-\alpha/2)] \ge  \sum_g f_g^{lim}(c^0)  = 1-\alpha/2.
\end{align}
Combining equation~\eqref{eq:limsup} and equation~\eqref{eq:liminf},
\begin{equation*}
  \lim_{n \rightarrow \infty} \mathbb{P}[ \sqrt{n} (\hat \theta_j^{(\bar g)} - \theta_{j}^{(0)} ) \le b^{-1}_{D_1,\ldots,D_n}(1-\alpha/2)]  = 1-\alpha/2.
\end{equation*}
Analogously, it can be shown that
\begin{equation*}
  \lim_{n \rightarrow \infty} \mathbb{P}[ \sqrt{n} (\hat \theta_j^{(\bar g)} - \theta_{j}^{(0)} ) \ge b^{-1}_{D_1,\ldots,D_n}(\alpha/2)]  = 1-\alpha/2.
\end{equation*}
Thus,
\begin{equation*}
  \lim_{n \rightarrow \infty} \mathbb{P}[ b^{-1}_{D_1,\ldots,D_n}(\alpha/2) \le \sqrt{n} (\hat \theta_j^{(\bar g)} - \theta_{j}^{(0)} ) \le b^{-1}_{D_1,\ldots,D_n}(1-\alpha/2)]  = 1-\alpha.
\end{equation*}
This completes the proof.

\end{proof}

\subsection{Proof of Theorem~\ref{thm:uniform}}

\begin{proof}
  First, with probability exceeding $1- \exp{-t}$,
  \begin{align*}
     & \, \,  \left| \frac{n}{d} \| \hat \theta^{(g)} - \hat \theta^{(0)} \|_2^2 -  \frac{n}{d}\| \delta^{(g)} \|_2^2 - \frac{1}{d} \| \tau^{(g)} \|_2^2 \right|  \\
    &=   \left| \frac{n}{d} \| \hat \theta^{(g)} - \hat \theta^{(0)} - \delta^{(g)} \|_2^2  - \frac{1}{d} \| \tau^{(g)} \|_2^2 + \frac{2n}{d} (\hat \theta^{(g)} - \hat \theta^{(0)} -\delta^{(g)}) \cdot   \delta^{(g)}    \right| \\
    & =  \left| \frac{n}{d} \| \epsilon^{(g)} - \epsilon^{(0)} \|_2^2  - \frac{1}{d} \| \tau^{(g)} \|_2^2 + \frac{2n}{d}    (\epsilon^{(g)} - \epsilon^{(0)}) \cdot   \delta^{(g)}    \right| \\
    &\le   C_1(c_0,\eta)  \frac{t}{d} +  C_2(c_0,\eta)\sqrt{\frac{t}{d}},
  \end{align*}
  for some constants $C_1(c_0,\eta)$ and $C_2(c_0,\eta)$. Here, we used sub-Gaussian and subexponential tail bounds, see for example Chapter~2 in \citet{wainwright2019high}. More precisely, we used that $ \frac{n}{d} (\epsilon^{(g)} - \epsilon^{(0)})\cdot \delta^{(g)}$ is sub-Gaussian with variance proxy %
  $\frac{4}{d} c_0^2 \max_{g,j} (\eta_j^{(g)})^2$ and that $ n\| \epsilon^{(g)} - \epsilon^{(0)} \|_2^2 - (\tau^{(g)})^2$ is subexponential with parameter $ 64 d \max_{g,j} (\eta_j^{(g)})^2$. %
  Using a union bound, for every $\kappa > 0$ there exists a constant $C_3(c_0,\eta,\kappa)$, such that with probability exceeding $1-\kappa$,
  \begin{equation*}
    \sup_g \left| \frac{n}{d} \| \hat \theta^{(g)} - \hat \theta^{(0)} \|_2^2 - \frac{n}{d} \| \delta^{(g)} \|_2^2 -  \frac{1}{d} \| \tau^{(g)} \|_2^2 \right| \le  C_3 \left(\sqrt{  \frac{\log G}{d}} + \frac{\log G}{d} \right)
  \end{equation*}
  Since $\log(G) /d \le c_1$, for all $\kappa > 0$ there exists a constant $C_4(c_0,c_1,\eta,\kappa)$ such that with probability exceeding $1-\kappa/2$,
  \begin{equation}\label{eq:firsthalf}
    \sup_g \left| \frac{n}{d} \| \hat \theta^{(g)} - \hat \theta^{(0)} \|_2^2 - \frac{n}{d}\| \delta^{(g)} \|_2^2 - \frac{1}{d}  \| \tau^{(g)} \|_2^2 \right| \le  C_4 \sqrt{  \frac{\log G}{d}} 
  \end{equation}
  Analogously, it can be shown that there exists a constant $C_5$ such that with probability exceeding $1-\kappa/2$,
  \begin{align}\label{eq:secondhalf}
    \begin{split}
    & \, \, \sup_g \left| \frac{n}{d} \| \hat \theta^{(g)} - \theta^{(0)} \|_2^2 - \frac{n}{d}\| \delta^{(g)} \|_2^2 -   \frac{1}{d}\| \sigma^{(g)} \|_2^2 \right|  \\
    &=  \sup_g \left| \frac{n}{d} \| \epsilon^{(g)} \|_2^2  - \frac{1}{d}\| \sigma^{(g)} \|_2^2 + \frac{2n}{d} \epsilon^{(g)} \cdot   \delta^{(g)}    \right| \\
    &\le C_5 \sqrt{ \frac{\log G}{d}}.
    \end{split}
  \end{align}
  Combining equation~\eqref{eq:firsthalf} and equation~\eqref{eq:secondhalf}, with probability $1-\kappa$,
  \begin{align*}
    & \sup_g \left| \frac{n}{d}  \| \hat \theta^{(g)} - \theta^{(0)} \|_2^2 - \frac{n}{d} \hat R^\text{mod}(g) \right| \\
    &\sup_g \left| \frac{n}{d}  \| \hat \theta^{(g)} - \theta^{(0)} \|_2^2 -  \max \left(\frac{n}{d} \| \hat \theta^{(g)} - \hat \theta^{(0)} \|_2^2 - \frac{1}{d} \sum_{j=1}^d  (\hat \tau_j^{(g)})^2 ,0 \right) - \frac{1}{d} \sum_{j=1}^d (\hat \sigma_{j}^{(g)})^2 \right| \\
    &\le   2\iota_{n,d} +  \sup_g \left| \frac{n}{d}  \| \hat \theta^{(g)} - \theta^{(0)} \|_2^2 -  \max \left(\frac{n}{d} \| \hat \theta^{(g)} - \hat \theta^{(0)} \|_2^2 - \frac{1}{d} \sum_{j=1}^d (\tau_j^{(g)})^2 ,0 \right) - \frac{1}{d} \sum_{j=1}^d (\sigma_j^{(g)})^2 \right| \\
    &\le 2\iota_{n,d} + C_6 \sqrt{\frac{\log G}{d}} + \sup_g \bigg| \frac{1}{d} \sum_{j=1}^d (\sigma_{j}^{(g)})^2 + \frac{n}{d} \|\theta^{(g)} - \theta^{(0)} \|_2^2   \\
    &- \max \left(\frac{n}{d} \| \theta^{(g)} - \theta^{(0)} \|_2^2 +  \frac{1}{d} \sum_{j=1}^d (\tau_j^{(g)})^2 - (\tau_j^{(g)})^2 ,0 \right) - \frac{1}{d} \sum_{j=1}^d (\sigma_{j}^{(g)})^2 \bigg| \\
    &\le 2\iota_{n,d} + C_6 \sqrt{\frac{\log G}{d}},
  \end{align*}
  for some constant $C_6$ that depends on $c_0$, $c_1$, $\eta$, and $\kappa$. Thus, on that event,
  \begin{equation*}
    \left| \min_g \frac{n}{d} \| \hat \theta^{(g)} - \theta^{(0)} \|_2^2 - \frac{n}{d} \min_g \hat R^\text{mod}(g) \right| \le 2 C_6 \sqrt{ \frac{\log G}{d}} + 4\iota_{n,d}.
  \end{equation*}
  Furthermore, on that event,
  \begin{equation*}
    \left| \frac{n}{d} \hat R^\text{mod}( \bar g) - \frac{n}{d} \| \hat \theta^{(\bar g)} - \theta^{(0)} \|_2^2 \right| \le C_6 \sqrt{ \frac{\log G}{d}} + 2\iota_{n,d}.
  \end{equation*}
  Thus,
  \begin{equation*}
    \frac{n}{d} \| \hat \theta^{(\bar g)} - \theta^{(0)} \|_2^2  \le  \min_g \frac{n}{d} \| \hat \theta^{(g)} - \theta^{(0)} \|_2^2 + 3 C_6 \sqrt{\frac{\log G}{d}} + 6 \iota_{n,d}.
  \end{equation*}
  Here, we used that $\min_g \hat R^\text{mod}(g) = \hat R^\text{mod}(\bar g)$.
   \end{proof}

\subsection{Auxiliary Lemmas}

\begin{lemma}\label{lemma:gauss2}
 Let $(Z_{i},X_{i})$ be i.i.d.\ Gaussian with mean zero and nonzero variance and $|\text{Cor}(Z_{i},X_{i})| \neq 1$. Let $K \ge 2$. %
 Then, 
\begin{equation}\label{eq:25}
  \text{Var} \left( \left( \frac{1}{K}\sum_{i=1}^{K} Z_{i} - X_{i} \right)^{2} \right) < \text{Var} \left( \frac{1}{K} \sum_{i=1}^{K} \left( \frac{1}{K-1} \sum_{j \neq i} Z_{j}  - X_{i} \right)^{2} \right).
\end{equation}
\end{lemma}
\begin{proof}
    As $(Z_{i},X_{i})$ are multivariate Gaussian and $|\text{Cor}(Z_{i},X_{i})| \neq 1$, $X_{i} = \alpha Z_{i} + \epsilon_{i}$, for some $\alpha \in \mathbb{R}$, where $(\epsilon_{i})_{i}$ is centered Gaussian and independent of $(Z_{i})_{i}$ with nonzero variance. Thus, it suffices to show that
\begin{equation}\label{eq:27}
  \text{Var} \left( \left( \frac{1}{K}\sum_{i=1}^{K} Z_{i} - \alpha Z_{i} - \epsilon_{i} \right)^{2} \right) < \text{Var} \left( \frac{1}{K} \sum_{i=1}^{K} \left( \frac{1}{K-1} \sum_{j \neq i} Z_{j}  - \alpha Z_{i} - \epsilon_{i} \right)^{2} \right).
\end{equation}
On the left-hand side of equation~\eqref{eq:27}, we have
\begin{align*}
  \begin{split}
    &   \text{Var} \left(  \left( \frac{1}{K}\sum_{i=1}^{K} Z_{i} - \alpha Z_{i} \right)^{2} -  2 \frac{1}{K}\sum_{i=1}^{K}  (Z_{i} - \alpha Z_{i}) \frac{1}{K} \sum_{i=1}^{K}  \epsilon_{i}  +  \left( \frac{1}{K}\sum_{i=1}^{K} \epsilon_{i} \right)^{2} \right) \\
    &= \text{Var} \left(  \left( \frac{1}{K}\sum_{i=1}^{K} Z_{i} - \alpha Z_{i} \right)^{2} \right) +  4  \text{Var} \left(\frac{1}{K}\sum_{i=1}^{K}  (Z_{i} - \alpha Z_{i}) \right) \text{Var} \left( \frac{1}{K} \sum_{i=1}^{K}  \epsilon_{i}  \right) \\
    &+  \text{Var} \left( \left( \frac{1}{K}\sum_{i=1}^{K} \epsilon_{i} \right)^{2} \right)
  \end{split}
\end{align*}
On the right-hand side of equation~\eqref{eq:27} we have \begin{align*}
  \begin{split}
    &\text{Var} \left( \frac{1}{K} \sum_{i=1}^{K} \left( \frac{1}{K-1} \sum_{j \neq i} Z_{j}  - \alpha Z_{i} \right)^{2}    -  \frac{2}{K} \sum_{i=1}^{K}  \epsilon_{i} \frac{1}{(K-1)} \sum_{j \neq i} Z_{j} - \alpha Z_{i}  +  \frac{1}{K} \sum_{i=1}^{K} \left(  \epsilon_{i} \right)^{2}\right) \\
    &= \text{Var} \left( \frac{1}{K} \sum_{i=1}^{K} \left( \frac{1}{K-1} \sum_{j \neq i} Z_{j}  - \alpha Z_{i} \right)^{2}  \right)  + 4 \text{Var} \left( \frac{1}{K} \sum_{i=1}^{K}  \epsilon_{i} \frac{1}{(K-1)} \sum_{j \neq i} Z_{j} - \alpha Z_{i} \right) \\
    &\qquad+  \text{Var} \left( \frac{1}{K}\sum_{i=1}^{K} \left( \epsilon_{i} \right)^{2}\right).
  \end{split}
\end{align*}
Thus, it suffices to show that the following three inequalities hold:
\begin{align*}
  \begin{split}
    \text{Var} \left( \left( \frac{1}{K} \sum_{i=1}^{K} Z_{i}  - \alpha Z_{i} \right)^{2} \right)   &\le  \text{Var} \left( \frac{1}{K} \sum_{i=1}^{K} \left( \frac{1}{K-1} \sum_{j \neq i} Z_{j}  - \alpha Z_{i} \right)^{2} \right) \\
    \text{Var} \left( \left( \frac{1}{K} \sum_{i=1}^{K} \epsilon_{i} \right)^{2}  \right)&< \text{Var}  \left( \frac{1}{K} \sum_{i=1}^{K} \left(  \epsilon_{i}  \right)^{2} \right) \\
  4  \text{Var} \left(\frac{1}{K}\sum_{i=1}^{K}  (Z_{i} - \alpha Z_{i}) \right) \text{Var} \left( \frac{1}{K} \sum_{i=1}^{K}  \epsilon_{i}  \right)   &\le   4 \text{Var} \left( \frac{1}{K} \sum_{i=1}^{K}  \epsilon_{i} \frac{1}{(K-1)} \sum_{j \neq i} Z_{j} - \alpha Z_{i} \right)
  \end{split}
\end{align*}
The first two inequalities follow from Lemma~\ref{lemma:baba} and Lemma~\ref{lemma:gaga}. Let us now deal with the last term.
\begin{align*}
  &\text{Var} \left( \frac{1}{K} \sum_{i=1}^{K}  \epsilon_{i} \frac{1}{(K-1)} \sum_{j \neq i} Z_{j} - \alpha Z_{i} \right) -    \text{Var} \left(\frac{1}{K}\sum_{i=1}^{K}  (Z_{i} - \alpha Z_{i}) \right) \text{Var} \left( \frac{1}{K} \sum_{i=1}^{K}  \epsilon_{i}  \right)  \\
  &= \sum_{i \neq j}\text{Var} \left( \frac{1}{K} \epsilon_{i} \frac{1}{(K-1)} Z_{j} \right) + \sum_{i} \text{Var} \left(  \frac{1}{K} \alpha \epsilon_{i} Z_{i} \right) - (1-\alpha)^{2} \frac{1}{K^{2}} \text{Var}(Z_{1}) \text{Var}(\epsilon_{1}) \\
  &= \text{Var}(\epsilon_{1}) \text{Var}(Z_{1}) \left( \frac{1}{K(K-1)} + \frac{\alpha^{2}}{K} - (1-\alpha)^{2} \frac{1}{K^{2}}  \right)
\end{align*}
Thus, it suffices to show that
\begin{equation*}
   \frac{1}{K(K-1)} - \frac{1}{K^{2}} + \frac{\alpha^{2}}{K} - \frac{\alpha^{2}}{K^{2}}  + 2 \alpha \frac{1}{K^{2}} \ge 0
 \end{equation*}
Or, equivalently,
 \begin{equation}\label{eq:33}
   \frac{1}{K^{2}(K-1)} + \frac{\alpha^{2} (K-1)}{K^{2}} + 2 \alpha \frac{1}{K^{2}} \ge 0.
 \end{equation}
Dividing by $(K-1)/K^{2}$ yields
 \begin{equation*}
   \frac{1}{(K-1)^{2}} + \alpha^{2} + 2 \frac{\alpha}{(K-1)} \ge 0.
 \end{equation*}
 The left term can be written as
\begin{equation*}
  \left( \frac{1}{(K-1) } + \alpha \right)^{2}.
\end{equation*} 
 This shows the inequality in equation~\eqref{eq:33}, which completes the proof.

\end{proof}

\begin{lemma}\label{lemma:baba}
  Let $\epsilon_{i}$ be i.i.d.\ centered Gaussian random variables with nonzero variance and $K \ge 2$. Then,
  \begin{equation}\label{eq:34}
  \text{Var} \left( \left( \frac{1}{K} \sum_{i} \epsilon_{i} \right)^{2}  \right) < \text{Var}  \left( \frac{1}{K} \sum_{i=1}^{K} \left( \epsilon_{i}  \right)^{2} \right)
\end{equation}  
\end{lemma}

\begin{proof}
  On the left-hand side of equation~\eqref{eq:34} we have
\begin{equation*}
  \frac{1}{K^{2}} \text{Var}(\epsilon_{1}^{2}).
\end{equation*}  
On the right-hand side of equation~\eqref{eq:34} we have
\begin{equation*}
  \frac{1}{K} \text{Var}(\epsilon_{1}^{2}).
\end{equation*}
This completes the proof.
\end{proof}

\begin{lemma}\label{lemma:gaga}
Let $Z_{i}$ be i.i.d.\ centered Gaussian random variables and $K \ge 2$. Then,
\begin{align*}
  \text{Var} \left( \left( \frac{1}{K} \sum_{i} Z_{i}  - \alpha Z_{i} \right)^{2} \right)   \le  \text{Var} \left( \frac{1}{K} \sum_{i=1}^{K} \left( \frac{1}{K-1} \sum_{j \neq i} Z_{j}  - \alpha Z_{i} \right)^{2} \right) .
\end{align*}
\end{lemma}

\begin{proof}
It suffices to show that
\begin{align}\label{eq:42}
  \text{Var} \left( \left( \frac{1}{\sqrt{K}}\sum_{i} Z_{i}  - \alpha Z_{i} \right)^{2} \right)   \le  \text{Var} \left( \sum_{i=1}^{K} \left( \frac{1}{K-1} \sum_{j \neq i} Z_{j}  - \alpha Z_{i} \right)^{2} \right).
\end{align}
On the left-hand side we have
\begin{equation*}
  (1-\alpha)^{4} \text{Var}(Z_{1}^{2}).
\end{equation*}
On the right-hand side of equation~\eqref{eq:42} we have
\begin{align*}
  & \text{Var} \left(  \frac{1}{(K-1)} \sum_{i}  Z_{i}^{2}  +  \frac{(K-2)}{(K-1)^{2}} \sum_{i > j}  2 Z_{i} Z_{j} - \sum_{i>j} \frac{\alpha}{(K-1)} 4 Z_{i} Z_{j}      + \alpha^{2} \sum_{i=1}^{K} Z_{i}^{2} \right) \\
  &=   K(\frac{1}{(K-1)} + \alpha^{2})^{2} \text{Var}(Z_{1}^{2}) +  (\frac{(K-2)}{(K-1)^{2}} - \frac{2 \alpha}{(K-1)})^{2} \frac{K (K-1)}{2} 2 \text{Var}(Z_{1}^{2}) \\
  &= \text{Var}(Z_{1}^{2}) (  K \alpha^{4}  - \alpha \frac{(K-2) 4 K}{(K-1)^{2} }  +   \alpha^{2} (2\frac{K}{(K-1)} + 4 \frac{K}{(K-1)}) \\
&+  \frac{K (K-1) + K (K-2)^{2}}{(K-1)^{3}} )
\end{align*}
Using the two inequalities above, it suffices to show that for all $\alpha$ and all $K \ge 2$,
\begin{align*}  
 & K \alpha^{4}   +  6\alpha^{2} \frac{K}{(K-1)} - \alpha \frac{(K-2)4K}{(K-1)^{2}}  +  \frac{K (K-1) + K (K-2)^{2}}{(K-1)^{3}} \\
  \ge&   \alpha^{4}  - 4 \alpha^{3} +  6 \alpha^{2} - 4 \alpha + 1
\end{align*}
Rearranging, it suffices to show that
\begin{align*}
  &(K-1) \alpha^{4}  + 4 \alpha^{3} +  6\alpha^{2} \frac{1}{(K-1)} +  \alpha \frac{4}{(K-1)^{2}}\\
  &+  \frac{K (K-1) + K (K-2)^{2} - (K-1)^{3}}{(K-1)^{3}}  \ge  0.
\end{align*}
Multiplying with $(K-1)$, this is equivalent to
\begin{align*}
  &(K-1)^{2} \alpha^{4} + 4 \alpha^{3} (K-1)  +     6\alpha^{2} \\
  &  +  \alpha \frac{4}{(K-1)} +  \frac{K (K-1) + K (K-2)^{2} - (K-1)^{3}}{(K-1)^{2}}  \ge 0,
\end{align*}
which is equivalent to
\begin{align*}
  &(K-1)^{2} \alpha^{4} + 4 \alpha^{3} (K-1) +      6\alpha^{2}  +  \alpha \frac{4}{(K-1)} \\
  & +  \frac{K^{2} -K + K^{3} - 4 K^{2}+ 4K  - K^{3} + 3K^{2} - 3K +1 }{(K-1)^{2}} \ge 0.
    \end{align*}
This inequality is equivalent to
\begin{equation}\label{eq:43}
(K-1)^{2} \alpha^{4} + 4 \alpha^{3} (K-1) +      6\alpha^{2}  +  \alpha \frac{4}{(K-1)} +  \frac{ 1 }{(K-1)^{2}}   \ge 0.
\end{equation}
Rearranging the left-hand side, we obtain
\begin{equation*}
 (K-1)^{2} \left( \alpha + \frac{1}{(K-1)} \right)^{4}.
\end{equation*}
This proves the inequality of equation~\eqref{eq:43}, which completes the proof.
    
\end{proof}

\end{appendix}

\end{document}